\documentclass[11pt]{article}

\usepackage{amsmath,amsthm, amssymb,amsfonts}
\usepackage{epsfig,cancel}

\numberwithin{equation}{section}

\newcommand{\mycaption}[1]{\caption{{\sf \small #1}}}

\newcommand{\beq}{\begin{equation}}
\newcommand{\eeq}{\end{equation}}
\newcommand{\ba}{\begin{array}}
\newcommand{\ea}{\end{array}}
\newcommand{\bea}{\begin{eqnarray}}
\newcommand{\eea}{\end{eqnarray}}
\newcommand{\bean}{\begin{eqnarray*}}
\newcommand{\eean}{\end{eqnarray*}}
\newcommand{\eref}[1]{(\ref{#1})}

\newcommand{\nn}{\nonumber}
\newcommand{\comment}[1]{}

\newcommand{\rk}{\mathop{{\rm rk}}}

\newcommand{\cO}{{\cal O}}
\newcommand{\cN}{{\cal N}}
\newcommand{\cF}{{\cal F}}
\newcommand{\cA}{{\cal A}}
\newcommand{\cB}{{\cal B}}
\newcommand{\cC}{{\cal C}}
\newcommand{\cL}{{\cal L}}
\newcommand{\cV}{{\cal V}}

\newcommand{\Hom}{{\rm Hom}}

\def\cjn1{{\cA, \cC^*\otimes \wedge^j \cN^*}}
\def\bjn1{{\cA, \cB^*\otimes \wedge^j \cN^*}}
\def\vjn1{{\cA, \cV^*\otimes \wedge^j \cN^*}}
\def\cjn2{{\cA, \cC\otimes \wedge^j \cN^*}}
\def\bjn2{{\cA, \cB\otimes \wedge^j \cN^*}}
\def\vjn2{{\cA, \cV\otimes \wedge^j \cN^*}}

\def\fnote#1#2{\begingroup\def\thefootnote{#1}\footnote{#2}
     \addtocounter{footnote}{-1}\endgroup}

\newtheorem*{conjecture}{Conjecture}
\newtheorem{lemma}{Lemma}
\newtheorem*{lemm}{Lemma I}
\newtheorem*{mma}{Lemma II}

\newcommand{\cK}{{\cal K}}

\begin{document}

\vspace{1cm}

\title{{\huge \bf
Stability Walls in Heterotic Theories
}}

\vspace{2cm}

\author{
Lara B. Anderson${}^{1,2}$,
James Gray${}^{3}$,
Andre Lukas${}^{3}$,
Burt Ovrut${}^{1,2}$
}
\date{}
\maketitle
\begin{center} {\small ${}^1${\it School of Natural Sciences,
      Institute for Advanced Study, \\ Princeton, New Jersey 08540,
      U.S.A.} \\ ${}^2${\it Department of Physics, University of
      Pennsylvania, \\ Philadelphia, PA 19104-6395, U.S.A.}
    \\${}^3${\it Rudolf Peierls Centre for Theoretical Physics, Oxford
      University,\\
      $~~~~~$ 1 Keble Road, Oxford, OX1 3NP, U.K.}\\
    \fnote{}{andlara@physics.upenn.edu, j.gray1@physics.ox.ac.uk, lukas@physics.ox.ac.uk, ovrut@elcapitan.hep.upenn.edu} }
\end{center}

\abstract{\noindent We study the sub-structure of the heterotic K\"ahler moduli space due to the presence of non-Abelian internal gauge fields from the perspective of the four-dimensional effective theory. Internal gauge fields can be supersymmetric in some regions of the K\"ahler moduli space but break supersymmetry in others. In the context of the four-dimensional theory, we investigate what happens when the K\"ahler moduli are changed from the supersymmetric to the non-supersymmetric region. Our results provide a low-energy description of supersymmetry breaking by internal gauge fields as well as a physical picture for the mathematical notion of bundle stability.  Specifically, we find that at the transition between the two regions an additional anomalous $U(1)$ symmetry appears under which some of the states in the low-energy theory acquire charges. We compute the associated D-term contribution to the four-dimensional potential which contains a K\"ahler-moduli dependent Fayet-Iliopoulos term and contributions from the charged states. We show that this D-term correctly reproduces the expected physics.  Several mathematical conclusions concerning vector bundle stability are drawn from our arguments. We also discuss possible physical applications of our results to heterotic model building and moduli stabilisation.
}

\newpage

\tableofcontents

%
%

\section{Introduction}

Heterotic compactifications on Calabi-Yau manifolds necessarily have
gauge field expectation values (vevs) in the internal dimensions.
This feature, which is a consequence of demanding that the total
charge for the Neveu-Schwarz form should vanish on the internal
compact space, gives rise to much of the complexity and structure of
these theories~\cite{Candelas:1985en,Green:1987mn}.  One of the interesting properties of these
gauge field vevs concerns their supersymmetry preserving
properties. The internal field strength, $F$, is usually chosen so as
to preserve ${\cal N}=1$ supersymmetry in the four-dimensional
effective theory. This can be imposed by demanding that the
ten-dimensional gaugino supersymmetry variations vanish. This leads to
the conditions \beq\label{HYME} g^{a\overline{b}}F_{\overline{b}a}=0\;
,\quad F_{ab}=F_{\overline{a}\overline{b}}=0 \eeq which are known as
the Hermitian Yang-Mills equations (here, $a$ and $\bar{b}$ are
holomorphic and anti-holomorphic indices on the Calabi-Yau,
respectively).  However, if one chooses gauge fields satisfying the
Hermitian Yang-Mills equations and, hence, preserving supersymmetry at
one point in moduli space, and then changes the values of the moduli,
one can find that Eqs.~\eqref{HYME} fail to have a solution and
supersymmetry becomes broken~\cite{duy1,duy2}. More specifically, it
is possible to demarcate regions in K\"ahler moduli space where the
gauge field vevs can preserve supersymmetry and regions where they
necessarily break it \cite{Sharpe:1998zu,Anderson:2009sw}. Given the explicit dependence of the
Eqs.~\eqref{HYME} on the metric, and hence the K\"ahler moduli, such a
behaviour is perhaps not surprising.

What happens in the effective field theory when the moduli evolve such
that the gauge fields break supersymmetry? One can see from the
dimensional reduction of the ten-dimensional effective action of
the $E_8 \times E_8$ heterotic theory that there will be a positive definite potential in the
non-supersymmetric region of K\"ahler moduli space. The argument goes as follows. Consider the
following three terms in the ten-dimensional effective action.
\begin{eqnarray}
\label{spartial}
S_{\textnormal{partial}} = -\frac{1}{2 \kappa_{10}^2} \frac{\alpha'}{4}\int_{{\cal M}_{10}} \sqrt{-g} \left\{ \textnormal{tr} (F^{(1)})^2+\textnormal{tr} (F^{(2)})^2 - \textnormal{tr} R^2 \right\}~.
\end{eqnarray}
The notation here is standard \cite{Green:1987mn} with the field
strengths $F^{(1)}$ and $F^{(2)}$ being associated to the two $E_8$
factors in the gauge group. One consequence of the ten-dimensional
Bianchi Identity,
\begin{eqnarray}
  dH = -\frac{3 \alpha'}{\sqrt{2}} \left( \textnormal{tr} F^{(1)} \wedge F^{(1)} +\textnormal{tr} F^{(2)} \wedge F^{(2)} - \textnormal{tr} R \wedge R \right) \;,
\end{eqnarray}
is its integrability condition,
\begin{eqnarray}
\label{intcond}
\int_{M_{6}} J \wedge \left( \textnormal{tr} \;F^{(1)} \wedge F^{(1)} + \textnormal{tr} \; F^{(2)} \wedge F^{(2)} - \textnormal{tr}\; R \wedge R \right) = 0 \;,
\end{eqnarray}
where $J$ is the K\"ahler form.  Now, suppose that we begin with a
supersymmetric field configuration, and then vary the K\"ahler moduli
while keeping the other moduli fixed.  The background gauge field
strengths are then $(1,1)$ forms to lowest order (as is the curvature
two-form). Using this observation, and the fact that we are working,
again to lowest order, with a Ricci flat metric on a manifold of
$SU(3)$ holonomy, equation \eqref{intcond} can be rewritten as
follows:
\begin{eqnarray}
  \int_{M_{10}} \sqrt{-g} \left( \textnormal{tr} (F^{(1)})^2+\textnormal{tr} (F^{(2)})^2 - \textnormal{tr} R^2  - \textnormal{tr} (F^{(1)}_{a \bar{b}} g^{a \bar{b}})^2   - \textnormal{tr} (F^{(2)}_{a \bar{b}} g^{a \bar{b}})^2 \right) = 0~.
\end{eqnarray}
Using this relation in \eqref{spartial}, we arrive at the following result
\begin{eqnarray}
\label{finalintro}
S_{\textnormal{partial}}= - \frac{1}{2 \kappa_{10}^2} \frac{\alpha'}{4} \int_{{\cal M}_{10}} \sqrt{-g} \left\{ \textnormal{tr} (F^{(1)}_{a \bar{b}} g^{a \bar{b}})^2  + \textnormal{tr} (F^{(2)}_{a \bar{b}} g^{a \bar{b}})^2 \right\}~.
\end{eqnarray}
The terms in Eq.~\eqref{finalintro} form a part of the ten-dimensional
theory which does not contain any four-dimensional derivatives. It
therefore contributes, upon dimensional reduction, to the potential of
the four dimensional theory. In the case of a supersymmetric field
configuration, the terms in the integrand of \eqref{finalintro}
vanish, these being precisely the squares of the first equation in
\eqref{HYME}.  Thus, in this case, no potential is generated. However,
if the K\"ahler moduli are varied such that the gauge field vevs are
no longer supersymmetric, \eqref{finalintro} no longer vanishes and we
obtain a positive definite contribution to the potential energy seen
in four dimensions. Thus, we are led to a picture of a perturbative
potential which, while positive definite in the non-supersymmetric
regions of moduli space, vanishes precisely where the gauge field vevs
preserve supersymmetry.

Beyond what is described above, it might seem difficult to write down
the exact expression for this potential in terms of the moduli
fields. Naively it seems like we need to know the metric and gauge
connection on the Calabi-Yau $3$-fold. These quantities are of course
unknown, except possibly numerically
\cite{DonaldsonNumerical,Douglas:2006hz,Douglas:2006rr,Braun:2007sn}. In
fact, however, one can analytically derive the exact form of this
potential as an explicit function of the moduli fields. This will be the main focus of the present paper.

Before we can discuss the explicit form of this potential, it is
useful to briefly review the mathematical language normally used to
describe supersymmetry within a heterotic compactification. The
question of whether a supersymmetric vacuum exists can be answered by
a mathematical analysis of the associated holomorphic vector bundle,
$V$, based on the Donaldson-Uhlenbeck-Yau theorem~\cite{duy1,duy2} and
the notion of slope-stability. We will explicitly carry this out later
in the paper. For the purpose of the present general discussion it
suffices to know that the supersymmetry properties of $V$ are governed
by a (maximally) destabilizing sub-bundle $\cF\subset V$ and a number
associated it, called the slope $\mu (\cF)$, which is a function of
the K\"ahler moduli, $t^i$, of the Calabi-Yau manifold. The vector
bundle $V$ is slope-stable and, hence, the associated gauge field is
supersymmetric, in the part of K\"ahler moduli space where $\mu (\cF
)<\mu(V)$, and it is unstable and supersymmetry is broken where $\mu (\cF)
>\mu(V)$. The boundary between those regions, defined by $\mu(\cF)=\mu(V)$,
divides the K\"ahler cone into regions of preserved and broken
supersymmetry. Such a co-dimension one ``boundary'' will be referred
to as a stability ``wall" in the K\"ahler cone. In the following
sections, we will demonstrate that, in fact, the potential given in
\eref{finalintro}, reproduces this structure.

\vspace{0.1cm}

We are now in a position to summarize our main results. For
concreteness, we will illustrate the structure of our effective theory
for a bundle $V$ with an internal gauge group $G=SU(3)$, but
analogous statements hold for other $SU(n)$ groups. In this case, $\mu(V)=0$. For a general
point in the supersymmetric region of the K\"ahler moduli space, that
is for $\mu (\cF)<0$, the low-energy gauge group is $E_6$ (times a
possible hidden gauge group which is not relevant to our discussion),
the commutant of $SU(3)$ within $E_8$. The matter field content
consists of a certain number of families and anti-families in ${\bf
  27}$ and $\overline{\bf 27}$ representations, respectively, plus a
number of singlet fields which can be interpreted as the bundle moduli
of $V$. For specific examples, the number of these multiplets can be
computed from the bundle cohomology of $V$ and we will do this later
in the paper. So far, this is simply the field content of a standard
heterotic Calabi-Yau compactification.

Next, we consider the theory at the stability wall, that is, the
boundary between supersymmetric and non-supersymmetric regions in the
K\"ahler cone where $\mu (\cF)=0$. Here, we find that the structure
group of the bundle ``degenerates'' to $S(U(2)\times U(1))$ and,
hence, the low-energy gauge group enhances from $E_6$ to $E_6\times
U(1)$. This theory has the same chiral asymmetry between ${\bf 27}$
and $\overline{\bf 27}$ multiplets as the theory at a generic
supersymmetric point in moduli space (although their individual
numbers may change), bundle moduli for the $S(U(2)\times U(1))$ bundle
and additional singlet fields $C^L$. The families/anti-families and
the $C^L$ fields carry a charge under the additional $U(1)$
symmetry. It is well-known, in the context of heterotic
compactifications~\cite{Dine:1986zy,Dine:1987xk,Lukas:1999nh}, that a
low-energy $U(1)$ symmetry which arises from of a $U(1)$ factor in the
internal gauge group is anomalous in the Green-Schwarz sense. The
$U(1)$ vector field is massive as a consequence of the Higgs
mechanism. In addition, associated to this $U(1)$ is a D-term which
contains a Fayet-Illiopolous (FI) contribution\footnote{We note that
  for internal gauge fields with structure group $G=U(1)$, it is
  known~\cite{Dine:1986zy,Dine:1987xk,
    Lukas:1999nh,Blumenhagen:2005ga} that Eq.~\eqref{finalintro} leads
  to a D-term potential associated with a Green-Schwarz anomalous $U(1)$
  symmetry~\cite{Dine:1987xk,Lukas:1999nh}. In fact, it is not
  difficult to derive this D-term potential from
  Eq.~\eqref{finalintro}. In the present paper, however, we are
  interested in the case of non-Abelian internal gauge groups,
  specifically $G=SU(n)$.}. In our case, we find that the U(1) D-term
takes the following form at lowest order in the expansions of
heterotic M-theory, and close to the boundary between the
supersymmetric and non-supersymmetric regions: \bea\label{FI-alpha}
D^{U(1)} = f(t^i) - \sum_{M ,\bar{N}} Q^M G_{M \bar{N}} C^M
\bar{C}^{\bar{N}}~.  \eea Here $G_{L\bar{M}}$ is a positive definite
metric and $Q^L$ are the $U(1)$ charges of the fields $C^L$. The FI
term, $f(t^i)$, takes the form (up to a positive constant of
proportionality)
\begin{equation}
 f(t^i) \sim \frac{\mu(\cF)}{\cV}
\end{equation}
with $\cV$ the Calabi-Yau volume and $\mu(\cF)$ is the slope parameter
(described above) of a sub-bundle. The associated D-term potential is the
explicit form of the potential described in equation
\eref{finalintro}.

Let us discuss this D-term \eref{FI-alpha} in the various regions of
the K\"ahler cone. At the stability wall, $\mu(\cF)=0$, the FI term
vanishes and, hence, the fields $C^M$ have a vanishing vacuum
expectation value. The combination of K\"ahler moduli
perpendicular to the stability wall receives a mass from the FI term
and represents the Higgs particle. Its axionic superpartner is
absorbed by the $U(1)$ vector field. All of the $C^L$ fields are
massless at the stability wall. Now we move into the region
$\mu(\cF)<0$ where supersymmetry should be preserved. In this region,
the FI term is negative and the fields $C^M$ develop a compensating
vev to set $D^{U(1)}=0$. Of course, this only works if there is at
least one negative $U(1)$ charge $Q^L$ and we will verify that this is
indeed the case. In this way, we find that supersymmetry is preserved
in the region $\mu(\cF)<0$, as expected. One might also ask about
matching the number of states we observe in this theory to the results
obtained from a standard analysis of the supersymmetric region. We
find that when the fields $C^L$ develop a vev, the $U(1)$ gauge boson
receives an additional contribution to its mass and eventually becomes
so massive that it should be dropped from the low-energy spectrum. In
this way, we recover the $E_6$ symmetry at a generic superymmetric
point. Further, due to the non-vanishing $C^L$ vevs, one combination
of fields, predominately made up from $C^L$ fields, now becomes the
Higgs multiplet and should be removed from the spectrum. For a
matching of states between the theory at the stability wall and at a
generic supersymmetric point in moduli space we need, therefore, that
the number of $S(U(2)\times U(1))$ bundle moduli plus the number of
$C^L$ fields equals the number of $SU(3)$ bundle moduli plus
one. Again, we will explicitly verify that this is true in
general. What happens if we move into the region $\mu(\cF)>0$ where we
expect supersymmetry to be broken? The above D-term will only lead to
broken supersymmetry in this region if the $C^L$ fields cannot
compensate for the, now positive, FI term. In other words, {\em all}
of the charges, $Q^L$, need to be negative if our D-term is to
reproduce the supersymmetry properties of the gauge bundle as derived
in higher dimensions. We will show that this is indeed always the
case.  In summary, the above D-term reproduces all of the expected
features of supersymmetry breaking induced by internal gauge fields, a
subject usually studied in the context of algebraic geometry. As such,
it provides a physical picture for the mathematical notion of slope
stability for vector bundles and it opens up a range of physical
applications, for example in relation to heterotic model building and
moduli stabilisation.

\vspace{0.2cm}

In the remainder of this paper, we derive the potential described
above, in detail, from first principles. In the next section, we
discuss the ten-dimensional picture, by introducing the mathematical
description of supersymmetric and non-supersymmetric gauge field vevs
in terms of vector bundles via the theorem of Donaldson, Uhlenbeck, and Yau
\cite{duy1,duy2}. We describe how one may study any given model to see
if it preserves or breaks supersymmetry at a given point in moduli
space. Section 3 uses this technology to show, from a ten-dimensional
perspective, how supersymmetric and non-supersymmetric regions, with
stability walls between them, arise in the K\"ahler cone.  In Section
4, we describe the four-dimensional effective description of this
phenomenon and derive the D-term~\eqref{FI-alpha}.  In Section 5, we
confirm the picture described in this introduction by studying the
vacuum space of the four-dimensional effective theory.  Higher order
corrections are explored in Section 6. In Section \ref{conclusions}, we
conclude and discuss further work. Certain mathematical details and a
conjecture are provided in Appendix A. In Appendix B, we provide another detailed example of a bundle exhibiting a stability wall in the K\"ahler cone and the explicit field theory describing it.

\section{Vector Bundle Stability in Heterotic
  Compactifications}\label{stability}
A supersymmetric heterotic string compactification requires the
geometric input of a complex three-dimensional Calabi-Yau manifold,
$X$, and a holomorphic vector bundle, $V$, defined over $X$. The gauge
connection, $A$, on $V$ with associated field strength, $F$, must
satisfy the Hermitian Yang-Mills equations~\eqref{HYME}.  On a
holomorphic vector bundle, $V$, one can always choose a connection with
a purely $(1,1)$ field strength, $F$, so that the last two conditions in
\eref{HYME}, $F_{ab}=F_{\overline{a}\overline{b}}=0$, are satisfied.
To solve the first equation~\eqref{HYME},
$g^{a\overline{b}}F_{\overline{b}a}=0$, is more difficult, at least
for the case of non-Abelian bundle structure groups. However, for
Calabi-Yau manifolds, there exists a powerful way of transforming this
equation into a problem in algebraic geometry. For K\"ahler manifolds,
the Donaldson-Uhlenbeck-Yau theorem \cite{duy1,duy2} states that on
each \textit{poly-stable} holomorphic vector bundle $V$, there exists
a unique connection satisfying the Hermitian Yang-Mills equation
\eref{HYME}. Thus, to verify that our vector bundle is consistent with
supersymmetry we need to verify that it possesses the property of
poly-stability.

The concept of stability of a bundle (or coherent sheaf), $\cF$, over a
K\"ahler three-fold, $X$, is defined by means of a quantity called the \textit{slope}:
\begin{equation}\label{slope}
\mu (\cF)\equiv \frac{1}{\rk(\cF)}\int_{X}c_{1}(\cF)\wedge J \wedge J \; . 
\end{equation}
Here, $J$ is the K\"{a}hler form on $X$, and ${\rm rk}(\cF)$ and
$c_1(\cF)$ are the rank and the first Chern class of $\cF$,
respectively. A bundle $V$ is now called stable (resp.~semi-stable) if
for all sub-sheaves $\cF\subset V$ with $0<{\rm rk}(\cF)<{\rm rk}(V)$
the slope satisfies \beq\label{slope_req} \mu (\cF)<\mu
(V)~~~~~(\text{resp.}~\mu (\cF)\leq \mu (V)) \;.  \eeq A bundle is
poly-stable if it can be decomposed into a direct sum of stable
bundles ($V=\bigoplus_n V_n$), which all have the same slope
($\mu(V_i)=\mu(V)$). It follows that every stable bundle is
poly-stable and, in turn, every poly-stable bundle is
semi-stable\footnote{Note that the converse to these statements do not
  hold. That is, not every semi-stable bundle is poly-stable,
  etc.}. Thus, as a series of implications: stable $\Rightarrow$
poly-stable $\Rightarrow$ semi-stable.

In this work, we will consider holomorphic vector bundles with
structure group $SU(n)$ with $n=3,4,5$. Since the slope of these
bundles vanishes($c_1(V) = 0$ for $SU(n)$ bundles), in order for
$V$ to be stable we must have that all proper sub-sheaves,
$\mathcal{F}$, of $V$ have strictly negative slope.  Thus if
$\mathcal{F} \subset V$ we require, \beq \mu(\mathcal{F}) < 0~.
\eeq\label{neg_slope} But what qualifies a sheaf $\mathcal{F}$ to be a
sub-sheaf of $V$? This is simply the condition that it has smaller rank
and that there exists an embedding $\mathcal{F} \hookrightarrow V$. The
space of homomorphisms between $\mathcal{F}$ and $V$, denoted
$\Hom_X(\mathcal{F},V)$, is isomorphic to the space of global
holomorphic sections $H^0(X, \mathcal{F}^* \otimes V)$. Hence, we have
that
\beq\label{hom} V ~{\rm stable } ~\Longleftrightarrow \mu(\mathcal{F}) <
0 ~~\forall \;\mathcal{F} ~{\rm s.t.}~ 0<\rk(\mathcal{F}) < n ~{\rm and }~
H^0(X, \mathcal{F}^* \otimes V) \ne 0 \ .  \eeq

To begin our study of stability, we will first re-write the slope
condition \eref{slope} into a form better suited to our
purposes. Given a basis of harmonic $(1,1)$ forms $J_i$ on $X$, where
$i,j=1,\ldots ,h^{1,1}(X)$, we expand the K\"ahler form as $J=t^iJ_i$
with the $t^i$ being the K\"ahler moduli. Inserting this into
Eq.~\eqref{slope}, the slope of a sheaf $\cF$ can then be written as
\begin{equation}
\mu(\mathcal{F}) =\frac{1}{rk(\mathcal{F})} d_{ijk}c_{1}^{i}(\mathcal{F})t^{j}t^{k}~, \label{t-slope}
\end{equation}
where the $d_{ijk}=\int_XJ_i\wedge J_j\wedge J_k$ are the triple
intersection numbers of $X$, and $c_1(\cF)=c_1^i(\cF)J_i$. It is useful
to define the ``dual K\"ahler moduli'', $s_i$, by \beq\label{sdef} s_{i}
\equiv d_{ijk}t^{j}t^{k} \; .  \eeq The slope then turns into
\beq\label{ss-slope} \mu(\mathcal{F})=\frac{1}{rk({\cal F})}
s_{i}c_{1}^{i}(\mathcal{F}) \; , \eeq and is, hence, given by a simple
dot product between the first Chern class of $\cF$ and the dual
K\"ahler moduli $s_i$. As stated above, we are interested in
bundles $V$ with structure group $G=SU(n)$ so that $c_1(V)=0$ and $\mu
(V)=0$.  Using Eq.~\eqref{ss-slope}, stability for such a bundle $V$
then amounts to the condition \beq\label{s-slope} \mu(\mathcal{F}) =
\frac{1}{rk({\cal F})} c_{1}^{i}(\mathcal{F})s_{i}<0 \; , \eeq for all
$\mathcal{F} \subset V$. Hence, for a given de-stabilizing sub-sheaf
$\cF\subset V$, Eq.~\eqref{s-slope} divides $s_i$ space into two
regions (the condition $\mu(\cF)=0$ defines a co-dimension $1$ hyperplane in $s_i$ space). To understand stability of a bundle, we need to analyse,
for all relevant sub-sheaves $\cF\subset V$, how these regions
relate to the K\"ahler cone (the allowed set of K\"ahler parameters
$s_i$). Concretely, this amounts to finding a region in the K\"ahler
cone which is not de-stabilised by any sub-sheaf $\cF\subset V$.

In the mathematics literature, when a bundle is called `stable' this
is frequently taken to mean stable with respect to \emph{all} possible
$s_{i}$. A choice of a vector $s_i$ is referred to as a `polarization'. That
is, the bundle which is stable with respect to all polarizations is
stable everywhere in the K\"ahler cone. However, viewed from the
perspective of physics, this is actually a stronger condition than we
require. In a heterotic compactification, we shall define our
low-energy effective theory perturbatively around a particular vacuum
corresponding to some point in moduli space. So, it is sufficient to show
that the bundle is stable \emph{somewhere} in the K\"ahler cone (with
the hope that we may eventually stabilize the moduli within
this region). In previous work \cite{stability_paper, Anderson:2008ex}, several of the
authors made use of this viewpoint to formulate stability criteria
for bundles defined over Calabi-Yau manifolds with $h^{1,1}(X)
>1$. The resulting algorithm is a generalization of the stability condition
given by Hoppe~\cite{hoppe} which, in its original form, applies to Calabi-Yau manifolds
with $h^{1,1}(X)=1$. In Ref.~\cite{stability_paper}, we describe this
algorithmic method for determining the stable regions of a
bundle. We will not repeat the details of this analysis here, but
rather highlight some of its important features in the following.

\subsection{Algorithmic testing of slope stability}
Vector bundle stability is a notoriously difficult property to
prove. The main obstacle arises in classifying all possible
sub-sheaves, $\cF$, of the bundle $V$. There are no general techniques
known for identifying such sub-sheaves or for computing their
topological properties such as Chern classes (and, hence, their slopes
from Eq.~\eref{slope}). However, as described in Ref.~\cite{stability_paper},
despite these obstacles, progress can be made by systematically
constraining the possible sub-sheaves, $\cF\subset V$.

\subsubsection{Sub-line bundles and stability}\label{line-bun}
In this section we will demonstrate that, in order to prove
stability at any given point in K\"ahler moduli space, it is
sufficient to test the slope criteria \eref{slope_req} for all
\emph{sub-line bundles} $\mathcal{L}$ of certain anti-symmetric powers, $\wedge^kV$,
of the bundle $V$.

To begin, consider a rank-$n$ vector bundle $V$ over a projective
variety $X$. If $\mathcal{F}$ is a sub-sheaf of $V$ then it injects
into $V$ via the resolution \beq\label{sub-sheaf} 0 \to \mathcal{F}
\to V \to \mathcal{K} \to 0~, \eeq with $\rk(\mathcal{F}) < \rk(V)$
and $\mathcal{K}=V/\cF$. We shall consider such sub-sheaves one rank
at a time. First, we observe that since $V$ is a vector bundle, it is
torsion-free and, thus, has no rank-zero sub-sheaves. So, we begin with
the case of a rank one sub-sheaf. Since $\mathcal{F}$ is torsion free,
there is an injection \beq\label{double-dual} \mathcal{F}
\stackrel{i}{\longrightarrow} \mathcal{F}^{**} \eeq where
$\mathcal{F}^{**}$ is the double-dual of $\mathcal{F}$. A locally free
coherent sheaf, $\mathcal{L}$, is isomorphic to its double-dual, that
is $\mathcal{L}^{**} \approx \mathcal{L}$. Since $\mathcal{F}$ is rank
one and torsion free, it can be shown that $\mathcal{F}^{**}$ is
locally free and, hence, a line bundle \cite{Friedmann}. Dualizing the
sequence \eref{sub-sheaf} twice and using \eref{double-dual} we have
\beq \mathcal{F} \subset \mathcal{F}^{**} \subset V^{**} \approx
V~. \eeq It is straightforward to show that
$\mu(\mathcal{F})=\mu(\mathcal{F}^{**})$. Thus, instead of checking
the slope condition \eref{slope_req} for all rank-one torsion-free
sub-sheaves of $V$, it suffices to check it for all
sub-line bundles. But what about sub-sheaves of higher rank?

Let $\mathcal{F}$ be a torsion free sub-sheaf  of rank $k$ (with
$1<k<n$). Once again, we have an inclusion $0 \to \mathcal{F} \to V$
which in turn induces a mapping \beq \wedge^{k}\mathcal{F}
\to\wedge^{k}V \eeq which can also be shown to be an injection
\cite{Donagi:2003tb}. By definition of the anti-symmetric tensor power
$\wedge^{k}$, $\wedge^{k}\mathcal{F}$ is a rank one sheaf.  Since
$\mathcal{F}$ is torsion free, so is $\wedge^{k}\mathcal{F}$
\cite{Kobayashi}. Next, by an argument similar to the one given above
(in and around \eqref{double-dual}), we can argue that there is a line
bundle $\cL$ associated to $\wedge^{k}\mathcal{F}$, namely
$\cL=(\wedge^k\cF)^{**}$.  Note that in general for a rank $n$ bundle
$V$, \beq\label{wedgec1} c_1(\wedge^{k}V)=\binom{n-1}{k-1}c_1(V)
\;.\eeq Thus we observe that for $SU(n)$ bundles, which have
$c_1(V)=0$, it follows that $c_1(\wedge^{k}V)=0$ as well. Likewise, we
see that applied to a rank $k$, sub-line bundle, $\mathcal{F}$,
\eref{wedgec1} gives us $\mu(\wedge^{k}\mathcal{F})
=\mu(\mathcal{F})$. Therefore, for each rank $k$ de-stabilizing
sub-sheaf of $V$ we have a corresponding de-stabilizing sub-line
bundle $\cL\subset\wedge^{k}V$ with the same slope as $\cF$. Thus in
proving stability of an $SU(n)$ vector bundle $V$, we need only show
that if $\mathcal{L} \subset \wedge^{k}V$, then \beq \mu(\mathcal{L})
< \mu(\wedge^{k}V)=0 \eeq for all $k$ with $0<k < n$. Since
line bundles are classified by their first Chern class on a Riemannian
manifold, this is a dramatic simplification of the problem of
stability. Rather than the untenable problem of considering all
sub-sheaves, we have only to analyze and constrain the well-defined
set of line bundle sub-sheaves of $\wedge^{k} V$.
\subsubsection{Constraints on line bundle sub-sheaves}\label{constl}
What constraints can we place on the line bundles which must be
considered in examining the stability of an $SU(n)$ bundle $V$?  Using
the results of the previous subsection, we begin by considering a line
bundle sub-sheaf $\mathcal{L}$ of $\wedge^{k}V$.  We present several
simple characteristics that distinguish line bundle sub-sheaves of
stable $SU(n)$ bundles.

First, as discussed in \eref{hom}, by definition, if $\mathcal{L}
\subset \wedge^{k}V$ then
\beq
 \Hom_{X}(\mathcal{L}, \wedge^{k}V) \neq 0~.
 \eeq
 Therefore, we have a non-trivial cohomology condition to
check for any candidate line bundle sub-sheaf\footnote{Note
  that $\Hom_{X}(\mathcal{L}, \wedge^{k}(V)) \neq 0$ implies that
  $\mathcal{L}$ is a line bundle sub-sheaf rather than a sub-line
  bundle of $V$. This follows from the fact that while injective maps
  exist, the image of $\mathcal{L}$ in $V$ may not be a
  bundle. Equivalently, it is possible that $V/\mathcal{L}$ is not
  always a bundle \cite{dan}.} of $V$. Note that in this section, we will consider the mapping of $\cL \hookrightarrow \wedge^{k}V$ for \textit{generic} values of the bundle moduli of $V$.
  
The second observation is that for $SU(n)$ bundles, if $V$ is stable
then $H^0(X,V)=H^0(X,V^*)=0$. Indeed, if $H^0(X,V)$ were non-vanishing, then it is clear that $\Hom_X (\cO,V)\cong H^0(X,\cO^*\otimes V)=H^0(X,V)\neq 0$ and, hence, that the trivial sheaf $\cO$ would de-stabilize $V$ for any choice of K\"ahler
moduli. A similar argument holds for $V^*$ which is stable exactly if $V$ is. For this reason, checking that $H^0(X,V)=H^0(X,V^*)=0$ for an $SU(n)$ bundle $V$ is a useful first test for stability which we can carry out before proceeding further. Assuming this has been verified, it is clear that all possible de-stabilizing line bundle sub-sheaves, $\cL \subset V$ (or $\cL\subset V^*$), must satisfy $H^0(X,\cL)=0$. Furthermore, if an $SU(n)$ bundle is stable then its anti-symmetric tensor powers, $\wedge^{k}V$, are at least semi-stable \cite{Friedmann, Kobayashi}. As a result, by scanning for possible line bundle sub-sheaves of $\wedge^{k}V$ for all values of $k$, we can definitively determine the region of stability of the $SU(n)$ bundle $V$. If we discover that for a fixed polarization, $\wedge^{k}V$ is destabilized by a line bundle $\cL$, then by the observations above, we know that $V$ itself is unstable for this K\"ahler form due to a rank $k$ sub-sheaf $\cF \subset V$.
 
To summarize, the method of analyzing the stability of an $SU(n)$
bundle at any given point in K\"ahler moduli space proceeds as
follows.
\vspace{0.2cm}

\begin{itemize}
\item {\bf Check that $H^0(X,V)=H^0(X,V^*)=0$.}

Should this not be the case, the bundle is unstable everywhere in K\"ahler cone and we can stop.

\item {\bf Consider all possible line bundles ${\bf \cal L}$, as
    classified by their first Chern class.}

  The results of the previous subsection assure us that we need only
  consider line bundles rather than all sheaves of rank $k<n$.

\item {\bf Discard all line bundles, $\cL$, for which ${\bf Hom(\cL, \wedge^{k} V)
    =0}$ for all ${\bf k<n}$.}

If $Hom(\cL, \wedge^{k} V)
    =0$, such a line bundle is not a sub sheaf of $\wedge^{k} V$ for any
$k<n$ and thus need not be considered\footnote{As an additional simplifying technique, we note that while scanning for possible line
bundle sub-sheaves of $\wedge^{k}V$, if it is true that $H^{0}(X,
\wedge^{k}V)=0~ \forall~k$, we can eliminate any line bundles for
which $H^0(X, \cL) \neq 0$.}. As a simplification, for $k=1,n-1$, we can discard all line bundles with $H^0(X,\cL)\neq 0$. Indeed, since we have already verified that $H^0(X,V)=H^0(X,V^*)=0$, such line bundles cannot inject into $V$ and $\wedge^{n-1}V\simeq V^*$.

\item {\bf Check the slope of the remaining line bundles}

  We must check the slope $\mu({\cal L})$ of the remaining line
  bundles at the point in K\"ahler moduli space we are considering. If
  $\mu({\cal L}) \geq \mu(V) =0$, then ${\cal L}$ destabilizes $V$ at that
  point in moduli space. If no such line bundle exists, then $V$ is slope-stable at this point in K\"ahler moduli space.
\end{itemize}
\section{Stability Walls in the K\"ahler Cone} \label{bundlestab}

The stability condition \eref{s-slope} clearly depends on the choice of K\"ahler
parameters and thus a bundle need not be stable throughout its
entire K\"ahler cone. Furthermore, the choice of bundle moduli can affect which
potentially de-stabilizing sub-sheaves inject into $V$. In principle then, ``walls" between regions of
stability/instability such as those depicted in the (dual) K\"ahler cone in Figure \ref{f:stab}
can occur. In the neighborhood of such stability walls, the supersymmetric
structure of the low energy effective theory must be studied in more
detail than in the stable region. We begin by exploring the structure of stability walls 
in K\"ahler moduli space.

While this discussion can be applied to a K\"ahler cone of any size, to illustrate this concept, we will consider a two-dimensional K\"ahler cone (that is, $h^{1,1}(X)=2$) given by the positive quadrant in the $(s_1,s_2)$ plane of dual K\"ahler moduli. Suppose that $V$ is an $SU(n)$ bundle $V$ and that a stability wall\footnote{In general, for an $h^{1,1}(X)$-dimensional K\"ahler moduli space, the stability wall will be a $(h^{1,1}(X)-1)$-dimensional hyperplane.} of the
form shown in Figure \ref{f:stab} is generated by a de-stabilizing
sub-sheaf $\cF \subset V$ with $c_1(\cF)=-kJ_1+mJ_2$, where $k>0$ and $m>0$~\footnote{Note that for an $h^{1,1}(X)$-dimensional, positive K\"ahler cone, if $\cF$ is to define a stability wall, $c_1^i(\cF)$ must contain at least one negative and one positive component. For this reason, a bundle defined on a manifold with $h^{1,1}(X)=1$ is stable everywhere or nowhere.}.  
\begin{figure}[!ht]
  \centerline{\epsfxsize=5in\epsfbox{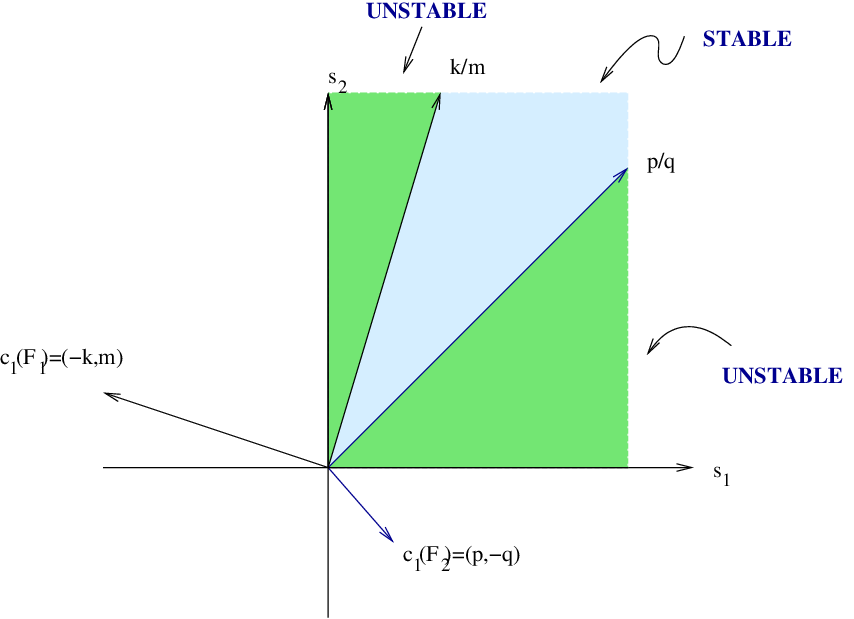}}
  \mycaption{ A two-dimensional dual K\"ahler cone, defined by $s_1 \geq 0$ and $s_2 \geq 0$, where $s_{i}=d_{ijk}t^{j}t^{k}$. Shown are two de-stabilizing sub-sheaves $\cF_1$ and $\cF_2$ with first Chern classes given by $c_1(\cF_1)=(-k,m)$ and $c_1(\cF_2)=(p,-q)$ for some integers $k,m,p,q$. The bundle $V$ is stable between the lines with slopes $k/m$ and $p/q$.}
\label{f:stab}
\end{figure}
From Eq.~\eqref{s-slope}, the slope of such a sub-sheaf is given by
\beq \mu(\cF)=\frac{1}{\rk({\cal F})} c_1^i(\cF)
s_i=\frac{1}{\rk({\cal F})} (-ks_{1}+ms_{2})~ . \label{2slope} \eeq
This means that, for all K\"ahler parameters $(s_1,s_2)$ where
$\mu(\cF)>0$, that is, for $s_{2}/s_{1}>k/m$, the bundle is unstable
while for $\mu(\cF)<0$, or $s_{2}/s_{1}<k/m$, it is potentially stable,
subject, of course, to other possible destabilizing sub-sheaves. For
example, in addition, there may exist a sub-sheaf with first Chern
class given by $c_1(\cF)=pJ_1-qJ_2$, where $p>0$ and $q>0$, which
would yield a lower boundary line with slope $p/q$. If these two
sub-sheaves are the ``maximally destabilizing'' ones on either side of
the K\"ahler cone, then the bundle is supersymmetric for all values
$p/q<s_{2}/s_{1}<k/m$. For a two-dimensional K\"ahler cone, the
supersymmetric region of a general bundle will be defined by these
upper and lower boundaries as illustrated in Figure \ref{f:stab}. For
the present discussion, we will focus our attention on the theory near
one of these boundary lines.

What happens on the line with slope $k/m$ itself? There, the bundle is
manifestly semi-stable since $\mu(\cF) =\mu(V)=0$. However, to decide
whether the low energy theory is supersymmetric or not, we must
consider not only our position in K\"ahler moduli space, but in bundle
moduli space as well.  If we examine this line in K\"ahler moduli
space while remaining at an arbitrary point in bundle moduli space for
which $V$ is an \emph{indecomposable} rank $n$ bundle, then
supersymmetry will be broken. This must be the case since
supersymmetric vacua exist if and only if the bundle is poly-stable. A
semi-stable bundle can only be poly-stable if it is a direct sum of
stable bundles. Therefore, the stability wall in K\"ahler moduli space
will only correspond to a supersymmetric solution if the bundle
decomposes into a direct sum $V \to \cF \oplus \cal{K}$ where
$\rk(\cF)+\rk({\cal{K}})= \rk(V)$ and $c_{1}(\cF)=
-c_{1}(\cal{K})$. Such bundle decompositions near a wall of
semi-stability were discussed for $K3$ manifolds in
Ref.~\cite{Sharpe:1998zu}.

At this special ``decomposable" locus in bundle moduli space, the
bundle is split and poly-stable. While the topological quantities of
$V$ remain the same at this locus, other important features of the
bundle and the corresponding low energy theory can change. For
instance, at this decomposable locus, the structure group of an $SU(n)$
bundle will become $S(U(n_1) \times U(n_2))$ with $n_{1}=\rk(\cF)$ and
$n_{2}=\rk(\cal{K})$. As we shall discuss in detail in the next
section, this change in the structure group of $V$ will also alter the
visible gauge symmetry of the four-dimensional theory. For instance,
if $\rk(V)=3$, then the commutant of $S(U(2) \times U(1))$ in $E_8$ is
no longer $E_6$, but is enhanced to $E_{6} \times U(1)$.

Before one can study such supersymmetric theories further, it is
prudent to ask whether such a decomposable point exists in the moduli
space of $V$. Fortunately, it can be shown that if there exists a
sub-sheaf $\cF$ of $V$ which injects into $V$, then there will always
exist a locus in the moduli space of $V$ for which $V$ decomposes as a
direct sum $\cF \oplus V/\cF$. If we define the relationship between
$\cF$ and $V$ via an `extension' short exact sequence 
\beq\label{extseq} 0 \to \cF\to V \to V/\cF \to 0~, 
\eeq then it is well-known that the space
of non-isomorphic extensions is given by ${\rm Ext}^{1}(V/\cF,\cF)$
\cite{AG}. Furthermore, the zero-element of the Ext group corresponds
to the decomposable locus $V=\cF \oplus {\cal K}$, where ${\cal
  K}=V/{\cF}$. If an indecomposable sub-sheaf $\cF$ defines a stability wall in moduli space, then, by definition, it must be stable\footnote{It is possible that $\cF$ could be a direct sum of stable objects with the same slope. In this case, we would simply obtain a further decomposition for $V$, that is, $V \sim \cF_1\oplus \cF_2 \oplus \ldots$.}. Simply dualizing the short exact sequence \eref{extseq}, $0 \to (V/\cF)^* \to V^* \to \ldots$ we see that $(V/\cF)^*$ must also be stable, since otherwise, it would destabilize $V$ in the region with $\mu(\cF)<0$, counter to our assumption of a stability wall. Thus, on a boundary, both $\cF$ and $V/\cF$ are stable sheaves.
 
 To see that the decomposition, $V=\cF \oplus V/\cF$ is given as a sum of stable bundles, rather than just stable torsion-free coherent sheaves, we must consider a Harder-Narasimhan filtration of $V$ \cite{Huybrechts}. More specifically, since $V$ is decomposable and semi-stable, we will consider a Jordan-H\"older filtration. Points on a moduli space of strictly semi-stable sheaves do not correspond to unique objects, rather they represent an ``S-Equivalence class"  \cite{Friedmann, Huybrechts, bradlow}. Two bundles are S-equivalent if their Jordan-H\"older graded sums ${\rm Gr}(V)=\cF_1\oplus V/\cF_1\oplus \ldots$ are isomorphic. For any S-equivalence class, there is a unique poly-stable representative up to isomorphism. That is, there is a unique graded sum in which the summands are stable bundles (i.e. a Seshadri Filtration, see \cite{Huybrechts, bradlow}) and, as a result, we can consider the decomposition of $V$ into this sum. 
 
 Thus, for bundles with non-trivial stable/unstable regions in
 K\"ahler moduli space, there will \emph{always} exist a locus in the
 moduli space of $V$ for which $V$ holomorphically decomposes as a
 direct sum of poly-stable bundles on the stability wall. We will now
 present a simple example of a Calabi-Yau manifold $X$ and a bundle
 $V$ which exhibits a stability wall.

\subsection{A stability wall example} \label{simpleeg}

Up to this point, our entire discussion has been completely general.
Let us now exemplify our previous comments by considering a bundle
defined on the complete intersection Calabi-Yau manifold~\cite{hubsch},
\begin{equation}
X= \left[\begin{array}[c]{c}\mathbb{P}^1\\\mathbb{P}^3\end{array}
\left|\begin{array}[c]{ccc}2 \\4
\end{array}
\right.  \right]\; ,
 \label{cicy24}
\end{equation}
defined by a polynomial of bi-degree $(2,4)$ in the ambient space $\mathbb{P}^1\times\mathbb{P}^3$. 
This manifold has two K\"ahler moduli, so $h^{1,1}(X)=2$. A basis of harmonic $(1,1)$ forms is given by the K\"ahler forms $J_1$ and $J_2$ of the ambient projective spaces $\mathbb{P}^1$ and $\mathbb{P}^3$ (pulled back to $X$). We denote the corresponding K\"ahler moduli by $t^1$ and $t^2$. The K\"ahler cone is the positive quadrant $t^1\geq 0$ and $t^2\geq 0$ and the non-zero triple intersection numbers are given by $d_{122}=4$ and $d_{222}=2$. From Eq.~\eqref{sdef}, we can calculate the dual K\"ahler moduli $s_1$ and $s_2$ and we find
\begin{equation}
 s_1=4(t^2)^2\; ,\quad s_2=8t^1t^2+2(t^2)^2\; . \label{sdefex}
\end{equation} 
Hence, expressed in terms of these dual K\"ahler moduli, the K\"ahler cone is the positive quadrant above the line $s_2/s_1=1/2$. Line bundles on $X$ are characterised by two integers, $k$ and $l$, and are denoted by $\cO_X(k,l)$. Their first Chern class is given by $c_1(\cO_X(k,l))=kJ_1+lJ_2$. 

We will define a rank $3$ monad
bundle \cite{Anderson:2007nc,Anderson:2008uw,Anderson:2008ex} on this
space by the short exact sequence \beq\label{monadeg} 0\to V \to
\cO_{X}(1,0)\oplus\cO_{X}(1,-1)\oplus\cO_{X}(0,1)^{\oplus 2}
\stackrel{f}{\longrightarrow} \cO_{X}(2,1) \to 0~.  \eeq The bundle $V$ is
defined as the kernel of the map $f$. This map is derived from
polynomials of bi-degree $((1,1),(1,2),(2,0),(2,0))$ (mapping sections
of $\cO_{X}(1,0)\oplus\cO_{X}(1,-1)\oplus\cO_{X}(0,1)^{\oplus 2}$ to sections of $
\cO_{X}(2,1)$). The rank of $V$ is three and $c_1(V)=0$ so that the structure group is generically $SU(3)$.
At a generic point in moduli space, the only non-vanishing
cohomology of this $SU(3)$ bundle is $h^1(X,V)=2$. This means there are two families in ${\bf 27}$ multiplets and no anti-families in $\overline{\bf 27}$. The moduli space of $V$ has dimension
\begin{equation}
h^1(X,V \otimes V^*)=22 \label{bundlemod},
\end{equation}
 so that we have $22$ $E_6$ singlet fields which should be interpreted as bundle moduli~\footnote{See \cite{Anderson:2008uw,yukawa} for general formulae for the spectra and moduli of monad bundles.}. We select the bundle~\eref{monadeg} solely because
it provides a straightforward example of a bundle with both
supersymmetric and non-supersymmetric regions in its moduli space,
and make no attempt here to consider models with fully realistic
particle spectra.

To analyze the stability of this rank three bundle we must consider
the potentially de-stabilizing rank one and two sub-sheaves. As discussed
in the previous subsections, this may be done more simply by
considering potentially de-stabilizing line bundle sub-sheaves of $V$
and $\wedge^2 V\cong V^*$.

Beginning with rank one sub-sheaves, we consider all sub-line bundles
of $V$. One can verify that $H^0(X,V)=0$ for the
bundle~\eqref{monadeg} and that all line bundles $\cO_X(k,l)$, where
$k,l\geq 0$ have sections. Hence, such semi-positive line bundles need
not be considered. Further, semi-negative line bundles $\cO_X(k,l)$,
where $k,l\leq0$ always have a negative slope in the interior of the
K\"ahler cone and are irrelevant.  It is, therefore, clear that the
only line bundles we need to consider are those with `mixed'
positive/negative entries in their first Chern classes. That is, $\cL$
is given by $\cO_X(-k,m)$ or $\cO_X(p,-q)$ for $k,m,p,q>0$. We seek
such line bundles for which $\Hom_X(\cL, V) \neq 0$. A straightforward
but lengthy analysis (see \cite{Anderson:2008ex, stability_paper} for
details) yields that if $\cL=\cO_X(-k,m)$ then $\Hom_X(\cL,V) \neq 0$
for $k \geq 3$ and $m=1$. Further, $\cO_X(p,-q)$ does not inject for
any values of $p,q$. Hence, the ``maximally destabilizing'' rank one
sub-sheaf corresponds to the line bundle $\cL_1=\cO_X(-3,1)$ and we
have the short exact sequence \beq 0 \to \cL_{1} \to V \to V/\cL_{1}
\to 0\; .  \eeq This implies that above a line with slope $s_2/s_1=3$
in the K\"ahler cone, the bundle is definitely unstable while it may be
stable below this line.

However, we still need to consider rank two destabilizing sub-sheaves or, equivalently, rank one line bundle sub-sheaves of $\wedge^{2}V$. As before, we find that no lower boundary exists, that is, $\Hom_X (O(p,-q), \wedge^{2}V)=0$ for all
values of $p,q>0$. For the upper boundary, we consider sub-bundles of
the form $\cL=\cO_X(-k,m)$ in $\wedge^2 V$. Since $V$ is an $SU(n)$
bundle we have $\wedge^2 V \simeq V^*$. This means we can extract information about $\wedge^2V$ from the dual
\beq
0 \to \cO_X(-2,-1) \to \cO_X(-1,0)\oplus \cO_X(-1,1) \oplus \cO_X(0,-1)^{\oplus 2} \to V^* \to 0
\eeq
of the monad sequence~\eqref{monadeg}. Twisting this sequence by $\cL^*=\cO_X(k,-m)$ we get
\beq
 0 \to \cO_X(k-2,-1-m) \to \cO_X(k-1,-m)\oplus \cO_X(k-1,1-m) \oplus \cO_X(k,-1-m)^{\oplus 2}\to \cL^*\otimes V^* \to 0~.
\eeq
One can verify from this sequence that $\Hom_X(\cL,\wedge^2V)\cong H^0(X,\cL^*\otimes V^*)\neq 0$ only for $\cL=\cO(-k,1)$ and $k \geq 1$. Hence, the maximally destabilizing line bundle is $\cL_2=\cO_X(-1,1)$ and we have 
\beq
  0 \to \cL_{2} \to V^* \to V^*/\cL_{2} \to 0\; .
\eeq 
Thus, $V^*$ is stable only below the line with slope
$s_2/s_1 =1$. Equivalently, this implies that there is a rank two sub-sheaf, ${\cal F}$ of $V$ with
$c_{1}({\cal F})=-J_1+J_2$ and
 \beq
  0 \to {\cal F} \to V \to V/{\cal F} \to 0~.
 \eeq
Since the rank two sub-sheaf ${\cal F}$ de-stabilizes a larger region of the
dual K\"ahler cone then the rank one sub-sheaf $\cL_{1}$, the existence
of $\cL_{1}=\cO_{X}(-3,1)$ is irrelevant here. While there are an infinite
number of sub-sheaves that de-stabilize some portion of the K\"ahler
cone, for this bundle there is only one relevant stability wall
which is determined by the rank two sub-sheaf ${\cal F} \subset V$. 
From Eq.~\eqref{2slope}, the slope of this sub-sheaf is given by
\begin{equation}
 \mu(\cF)=\frac{1}{2}(-s_1+s_2)\; , \label{slopeex}
\end{equation}
and it follows that $V$ is stable below the line with slope $s_2/s_1=1$. 
The dual K\"ahler cone, together with the region of stability, are plotted in
Figure \ref{f:example7887}.
\begin{figure}[!ht]
\centerline{\epsfxsize=5in\epsfbox{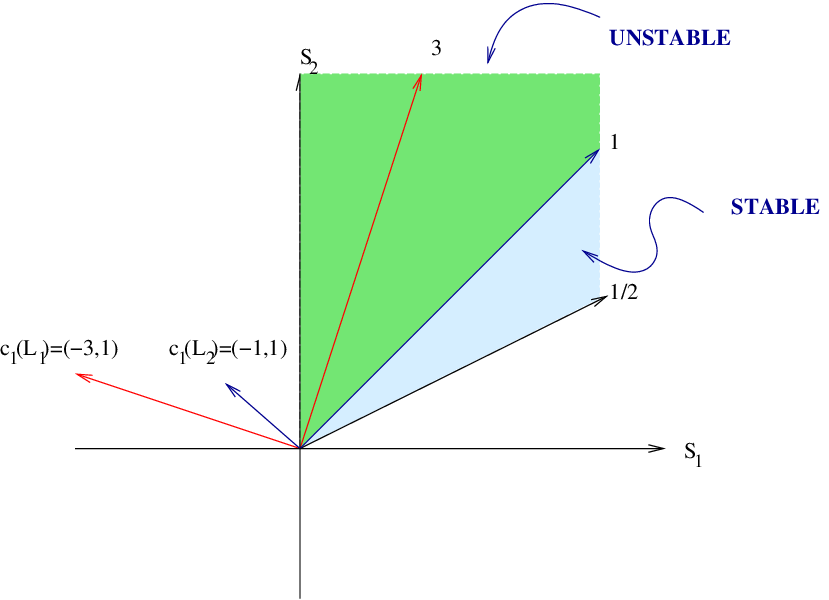}}
\mycaption{ The dual K\"ahler cone and the regions of
    stability/instability for the monad bundle described in Section
    \ref{simpleeg}. Here $L_1$ and $L_2$ are line bundle sub-sheaves of $V$ and $\wedge^2(V)$ respectively. The boundaries of the dual K\"ahler cone are denoted by the $s_2$-axis and the line with slope $1/2$.}
\label{f:example7887}
\end{figure}

The discussion of the last two sections is rather mathematical in
nature. It would be desirable to have more physical insight into what
is going on, and to be able to describe stability walls in the K\"ahler
cone in terms of the four-dimensional effective action. To this end,
in the next section, we will study the effective four-dimensional
theory describing fluctuations about the stability wall, that is, the
locus in moduli space where the bundle structure group decomposes. We
will then use these results, and the ones of the present section, to
discuss what happens physically as one crosses a wall of stability.

\section{Effective Field Theory at the Decomposable Locus}
\label{eft}

In this section, we will compute the potential in the four-dimensional
effective theory near the locus in bundle moduli space where the
sequence $ 0 \to {\cal F} \to V \to {\cal K} \to 0 $ becomes the
trivial extension, that is, where the bundle decomposes as $V = {\cal
  F} \oplus {\cal K}$. To perform such a computation, the first thing
we need to know is the low energy spectrum. We will describe this in
two stages; first presenting those fields which descend from
ten-dimensional gauge fields before continuing to describe those which
arise from other sources.

\subsection{Four-dimensional spectrum from the gauge
  sector}\label{4dspec1}

As stated in the previous section, it is clear from the sequence $ 0
\to {\cal F} \to V \to {\cal K} \to 0 $, and the fact that $c_1(V)$
vanishes, that ${\cal F}$ and ${\cal K}$ have equal and opposite first
Chern class. Thus, at the decomposable point in bundle moduli space
the structure group of $V$ is $S (U(n_1) \times U(n_2) )$, where $n_1+n_2=n$. It will turn
out that this is not the most convenient way in which to express this
group for what follows, in particular for the calculation of the
spectrum. As such, we will now carry out a little bit of group theory
in order to obtain a more suitable form.

Locally, at the level of Lie algebras, $S(U(n_1) \times U(n_2))$ is
equivalent to $SU(n_1) \times SU(n_2) \times U(1)$. Elements of the
former group are defined by a pair $(A,B)$, where $A$ and $B$ are $n_1
\times n_1$ and $n_2 \times n_2$ unitary matrices respectively,
satisfying the condition $\det{A} \det{B}=1$. Elements of $SU(n_1)
\times SU(n_2) \times U(1)$ are defined by a triplet $({\cal A}, {\cal
  B}, {\cal E})$, where ${\cal A}$ and ${\cal B}$ are $n_1 \times n_1$
and $n_2 \times n_2$ special unitary matrices respectively and ${\cal
  E}$ is the $U(1)$ phase. We may define a map $\sigma: SU(n_1)\times
SU(n_2)\times U(1)\rightarrow S(U(n_1)\times U(n_2))$ by $({\cal
  A},{\cal B},{\cal E})\rightarrow (A,B)= ({\cal E}^{n_2} {\cal
  A},({\cal E}^*)^{n_1} {\cal B})$ and it is easy to verify that this
map is onto and that ${\rm Ker}(\sigma
)\cong\mathbb{Z}_{n_1n_2}$. Hence, globally $S(U(n_1)\times
U(n_2))\cong (SU(n_1)\times SU(n_2)\times U(1))/\mathbb{Z}_{n_1n_2}$.
To understand the matter content of the low energy heterotic theory we
must consider the branching of the adjoint of $E_8$ under the bundle
structure group and its commutant. In the standard texts
\cite{Slansky:1981yr}, these branchings are given in terms of $SU(n_1)
\times SU(n_2) \times U(1)$ rather than $S(U(n_1) \times U(n_2))$
which is why we have discussed the relation between those two groups.

For the sake of brevity we will only consider one possible structure
group in the main text of this paper. We shall detail in full the case
$SU(3)\rightarrow S(U(2) \times U(1))$ and note that all other $SU(n)$ decompositions
follow in an entirely analogous manner\footnote{For example, there are two possible
  decompositions for an $SU(3)$ bundle. First, we have $SU(3)
  \rightarrow S(U(2)\times U(1))\approx SU(2)\times U(1)$, corresponding to $V \rightarrow \cF
  \oplus \cal{K}$, a sum of a rank two and a rank one bundle. There is
  a second possibility, namely $SU(3) \rightarrow S(U(1) \times U(1) \times
  U(1)) \approx U(1) \times U(1)$, corresponding to a decomposition
  into three line bundles: $V \rightarrow \cL_1 \oplus \cL_2 \oplus
  \cL_3$. In this latter case, one would find two additional low energy
  $U(1)$ symmetries. In the interests of brevity, we will only detail the case
  of a single $U(1)$ here.}. We consider, then, the case where we have an $SU(3)$ structure group
at a generic point in moduli space, degenerating to $SU(2) \times
U(1)$ at the stability wall about which we construct our low
energy theory. This gives us a low energy gauge group $E_6
\times U(1)$ at this locus. Under the decomposition $E_8
\supset E_6 \times SU(2) \times U(1)$ the adjoint of $E_8$ decomposes
as follows.
\begin{eqnarray}
\label{decomp1}
{\bf 248} &=& ({\bf 1},{\bf 1})_0 + ({\bf 1},{\bf 2})_{-3} +({\bf 1},{\bf 2})_{3} +({\bf 1},{\bf 3})_0 + ({\bf 78},{\bf 1})_0 \\ \nonumber
&&+({\bf 27},{\bf 1})_{2}+({\bf 27},{\bf 2})_{-1}+(\overline{\bf 27},{\bf 1})_{-2} + (\overline{\bf 27},{\bf 2})_{1}
\end{eqnarray}
In the above decomposition, the first number in the bracket is the
$E_6$ representation, the second number is that of a $SU(2)$
representation and the subscript is the $U(1)$ charge. We note that
our sign conventions differ somewhat from those of
\cite{Slansky:1981yr}.

The field content of the low energy theory is determined by the first
and zeroth cohomologies of various combinations of ${\cal F}$ and ${\cal K}$ as determined by the decomposition \eqref{decomp1}. The
first cohomologies tell us about scalars and the zeroth about gauge
bosons in the four-dimensional effective theory.  We must remember that
the groups $SU(2)$ and $U(1)$ in the above branching are not directly the structure groups
of $\cF$ and $\cK$ in the decomposition $V=\cF\oplus \cK$. Rather, since $\cF$ is a rank two
bundle with non-vanishing first Chern class its structure group is $U(2)$.  
Further, the structure group of $\mathcal{K}$ is $U(1)$ with
the additional constraint that $c_{1}(\cF) + c_1(\mathcal{K})=0$
\footnote{Note that we have assumed here that it is the rank 2
  sub-sheaf which injects everywhere and destabilizes $V$. The same
  analysis can be repeated assuming that it is the rank 1 sub-sheaf
  which is destabilizing without changing the result. This is because
  the only information which will enter the considerations of this
  subsection is the nature of the bundle {\it at the decomposable
    locus in bundle moduli space}.}, so that the overall structure group of $\cF\oplus \cK$ is $S(U(2)\times U(1))$.
 The proceeding group theory discussion tells us that the elements of this structure group
 are given by $({\cal E A},({\cal E}^*)^2)$ where $({\cal A},{\cal E})\in SU(2)\times U(1)$. 
 We have summarised the information about the various representations and cohomologies, associated to low-energy chiral multiplets, in Table~\ref {table1}. Note that the charges given as a subscript in the first column
refer to the $U(1)\subset SU(2)\times U(1)$ while the charges in the last column refer to the $U(1)$ in the commutant of $S(U(2) \times U(1))$ in $E_8$.
\begin{table}
\begin{center}
\begin{tabular}{|c|c|c|}
\hline
Representation &Cohomology& Physical $U(1)$ charge\\ \hline
$({\bf 1},{\bf 2})_{-3}$ &  $H^1(X, {\cal F} \otimes {\cal K}^*)$ & $-3/2$ \\ \hline
$({\bf 1},{\bf 2})_{3}$ & $ H^1(X, {\cal F}^* \otimes {\cal K})$ & $3/2$\\ \hline
$({\bf 1},{\bf 3})_0$ &  $H^1(X, {\cal F} \otimes {\cal F}^*)$ & $0$ \\ \hline
$({\bf 27},{\bf 1})_{2}$ & $H^1(X, {\cal K})$ & $1$ \\ \hline
$({\bf 27},{\bf 2})_{-1}$ & $H^1(X, {\cal F})$ & $-1/2$\\ \hline
$(\overline{\bf 27},{\bf 1})_{-2}$ & $H^1(X, {\cal K}^*)$ & $-1$ \\ \hline
$(\overline{\bf 27},{\bf 2})_{1}$ & $H^1(X, {\cal F}^*)$ & $1/2$\\ \hline
\end{tabular}
\mycaption{Representations, cohomologies and $U(1)$ charges associated to the zero modes which arise at the stability wall. The first column gives the representation under $E_6\times SU(2)\times U(1)$, the second column is the relevant cohomology involving $\cF$ and $\cK$, and the last column is the charge of the states under the $U(1)$ which is in the commutant of $S(U(2) \times U(1))$ in $E_8$.}
\label{table1}
\end{center}
\end{table}
Let us interpret the fields that appear here carefully. The $({\bf
  27},{\bf 1})_{2}$, $(\overline{\bf 27},{\bf 1})_{-2}$, $({\bf
  27},{\bf 2})_{-1}$ and $(\overline{\bf 27},{\bf 2})_{1}$ multiplets
unambiguously represent matter fields while the $({\bf 1},{\bf 3})_0$
multiplet clearly corresponds to moduli of the $S(U(2)\times U(1))$
bundle. The remaining two cohomologies, however, are a little bit more
subtle to interpret. From the point of view of the theory at stability
wall - the point of view we are considering here - these fields are
charged under a visible sector gauge group (the enhanced $U(1)$) and,
hence, they are matter fields. However, it would also not be
unreasonable to regard them as bundle moduli.  In general, we think of
the cohomology $H^1(X,V \otimes V^*)$ as representing bundle
moduli~\footnote{More precisely, as a vector space, $H^1(X, V \otimes
  V^*)$ can be viewed as the tangent space to bundle moduli
  space.}. At the stability wall, where the bundle $V$ decomposes as $V
= {\cal F} \oplus {\cal K}$, this bundle cohomology splits into various
parts as \bea\label{badger_moduli} H^1(X, V \otimes V^*) = H^1(X,
{\cal F}\otimes {\cal F}^*) \oplus H^1(X, {\cal F}^* \otimes {\cal K})
\oplus H^1(X, {\cal F} \otimes {\cal K}^*)\; .  \eea Here, we have
used that ${\cal K}$ is a line bundle in the case we are considering
and that $H^1(X, {\cal O})=0$ on a Calabi-Yau manifold. Thus, it is
not unreasonable to interpret $H^1(X, {\cal F}^* \otimes {\cal K})$
and $H^1(X, {\cal F} \otimes {\cal K}^*)$ as giving rise to bundle
moduli. Thinking about the perturbations such degrees of freedom would
contribute to the higher dimensional gauge field, we see that they
describe the deformations of the split bundle where ${\cal F}$ and
${\cal K}$ are mixed into one another; that is, they parametrize
movement in moduli space away from the decomposable locus. However, we
stress that, in this work, the effective field theory that we will
derive will describe perturbations around the decomposable locus, and
hence, we will think of the charged fields in $H^1(X, {\cal F}^*
\otimes {\cal K})$ and $H^1(X, {\cal F} \otimes {\cal K}^*)$ as
matter. In the following, these fields will be denoted by $C^L$.

\subsection{Four dimensional spectrum from the gravitational
  sector}\label{dim_reduction}

In addition to the fields of the previous subsection, we have the
usual low energy moduli from the gravitational sector of the
eleven-dimensional theory. It is important to note that some of these
moduli are also charged under the $U(1)$ symmetry in the low energy
gauge group even though they do not descend from higher-dimensional
gauge fields. The moduli fields which are not associated to the
$E_8\times E_8$ gauge fields include the dilaton, the (complexified)
K\"ahler moduli, the complex structure moduli and possible five-brane
moduli. It turns out that, of these fields, only the complex structure
moduli are not charged under the $U(1)$ symmetry. For now, we will
focus on tree level results where only the K\"ahler moduli are of
importance. The other fields will come into play in Section 6 where we
calculate what, in the weakly coupled language, correspond to one-loop
corrections. For reasons which will become clear, the following
arguments will be carried out using the language of the strongly-coupled
$E_8\times E_8$ heterotic
string~\cite{Witten:1996mz}-\cite{Lukas:1998tt}, that is, M-theory on
the orbifold $S^1/\mathbb{Z}_2$. However, analogous arguments leading
to the same results can be presented starting with the weakly-coupled
ten-dimensional theory~\cite{Green:1987mn}.

In terms of higher-dimensional fields, the K\"ahler moduli $T^i$, where $i,j,k=1,\ldots ,h^{1,1}(X)$ can be written as follows.
\begin{equation}
 T^i=t^i+2i\chi^i \label{Tdef}
\end{equation}
Here $t^i$ are the K\"ahler parameters of the Calabi-Yau manifold,
which we have already encountered in our bundle stability analysis,
and $\chi^i$ are the associated $T$-axions which descend from the
M-theory three-form as
 \bea \label{axdef} C_{11 a \bar{b}} = \chi^i
J_{i a \bar{b}}\; .  
\eea 
We recall that $\{J_i\}$ is a basis of
harmonic $(1,1)$ forms on the Calabi-Yau manifold, chosen to be dual
to a basis $\{{\cal{C}}^i\}$ of the second Calabi-Yau homology such that
\begin{equation}
 \frac{1}{v^{1/3}}\int_{{\cal{C}}^i}J_j=\delta^i_j\; ,
\end{equation} 
where $v$ is an arbitrary coordinate volume of the Calabi-Yau space.
The index $11$ refers to the coordinate of the $S^1/\mathbb{Z}_2$ orbifold, and $a,b,\dots$ and $\bar{a},\bar{b},\dots$ denote holomorphic and anti-holomorphic Calabi-Yau indices. 

It is a well-known fact that anti-symmetric tensor fields in heterotic theories
transform under $E_8\times E_8$ gauge transformations \cite{Green:1987mn}.  Consider
a local infinitesimal gauge transformation, 
\bea \label{gt1} \delta
A_A = - D_A \epsilon \;, 
\eea 
where the derivative is covariant and
$\epsilon$ is the gauge transformation parameter. Under such a change
of gauge the two-form $C_{11AB}$ transforms as 
\bea \label{gchange} \delta C_{11AB} = - \left(\frac{\kappa_{11}}{4 \pi}\right)^{2/3} \frac{1}{4 \pi}
\delta(x^{11}) \textnormal{tr}( \epsilon F_{AB}) \; ,
\eea
where $A,B =0,\ldots,9$ label the coordinates transverse to the $S^1/\mathbb{Z}_2$ orbifold.
Let us concentrate first on the internal components of equation \eqref{gchange} by
writing $\delta C_{11 a \bar{b}} = \; \delta \chi^i J_{i a\bar{b}}$. 
Integration over ${\cal C}^i \times S^1/Z_2$ then leads to the following gauge transformation
\bea \label{chichange}
 \delta \chi^i = - \frac{\epsilon_S \epsilon_R^2}{16\pi} \int_{{\cal C}^i} \textnormal{tr}( \epsilon F )\eea
for the $T$-axions, where we have introduced the dimensionless $\cO (\kappa_{11}^{2/3})$ combination of constants
\begin{equation} \label{combodef}
  \epsilon_S \epsilon_R^2  = \left(\frac{\kappa_{11}}{4 \pi}\right)^{2/3} \frac{8}{ 4 \pi \rho v^{1/3}}
\end{equation} 
and $\pi\rho$ is the coordinate volume of the $S^1/\mathbb{Z}_2$
interval. The constants $\epsilon_S$ and $\epsilon_R$ are the usual
expansion parameters defining four-dimensional heterotic M-theory
\cite{Lukas:1998hk}.  In weakly coupled language, $ \epsilon_s
\epsilon_R^2=8 \pi\alpha '/(4 v_{10})^{1/3}$, where $v_{10}$ is the
Calabi-Yau coordinate volume in the 10-dimensional theory.  Hence, the
effects considered here are order $\alpha'$ but at tree
level. Corrections which are one-loop from a weakly coupled
perspective will be discussed in Section~\ref{higherorder}. Normally,
the transformation~\eqref{chichange} does not lead to a non-trivial
gauge transformation of the $T$-axions under the visible sector gauge
group. This is because if $F$ is nonzero, in order for $\chi$ to have
a non-trivial transformation according to \eqref{chichange}, then $F$
breaks the associated gauge symmetry at the compactification scale and
so it does not appear as a factor in the visible sector gauge
group. However, in our case we have a $U(1)$ factor in the structure
group of our bundle which, due to its self commutation, is both
visible and hidden at the same time. In particular, $F$ can have a
non-trivial vev in the $U(1)$ direction without breaking the
associated visible sector gauge symmetry.

Thus, for our case, we have a non-trivial $U(1)$ transformation for the
moduli $T^k$. If we consider a gauge transformation associated with the
additional $U(1)$ seen in the visible sector, with a gauge parameter denoted by $\tilde{\epsilon}$,
we may rewrite \eqref{chichange} in terms of the first Chern $c_1({\cal F})$ of the de-stabilizing sub-sheaf ${\cal F}$ as follows.
\bea
\label{chiX}
\delta \chi^i = - \frac{3}{16} \epsilon_S \epsilon_R^2 \tilde{\epsilon}\; c_1^i({\cal F})
\eea
In addition, the singlet matter fields $C^L$ carry a $U(1)$ charge $Q^L$ and transform linearly as
\begin{equation}
 \delta C^L= -i\tilde{\epsilon} Q^L C^L \; . \label{CX}
\end{equation}
It is known~\cite{Dine:1986zy,Dine:1987xk,Distler:1987ee} that a low-energy $U(1)$ symmetry in heterotic compactifications which arises as the commutant of a $U(1)$ factor in the internal bundle structure group is generally anomalous in the Green-Schwarz sense. In this case, the triangle anomalies in the four-dimensional theory are cancelled by an anomalous variation of the gauge kinetic function, as usual. Perhaps not so well-known, but explained in detail in Ref.~\cite{Lukas:1999nh}, is that this includes an anomalous variation of the $T$-modulus dependent threshold corrections of the gauge kinetic function which transforms under~\eqref{chiX}. We stress that this is different from the perhaps better known ``universal'' anomalous $U(1)$ where the triangle anomaly is cancelled by a variation of the dilaton only~\cite{Dine:1986zy,Distler:1987ee}. This universal anomaly arises when the anomalous $U(1)$ symmetry has no internal counterpart in the bundle structure group. Such a situation can arise in the $SO(32)$ heterotic string but not in smooth compactifications of the $E_8\times E_8$ theory. Hence, in the present context we are always dealing with a ``non-universal'' anomalous $U(1)$ symmetry which transforms the $T$-moduli as in Eq.~\eqref{chiX}. In general, amomalous $U(1)$ symmetries are associated to FI terms in the four-dimensional theory. While the universal heterotic $U(1)$ implies the well-known dilaton-dependent FI term~\cite{Dine:1986zy,Distler:1987ee}, the present non-universal case leads to a $T$-dependent FI term~\cite{Lukas:1999nh} at leading order. We will now derive this FI term explicitly.

\subsection{The four dimensional potential}

The potential of an ${\cal N}=1$ supersymmetric theory contains two
types of contribution: those from D and F-terms. As we will see later, F-terms
are less relevant in our context, so we focus on D-terms and their associated potential. 
To do this, we need to know the K\"ahler potential of the fields involved. Having deferred loop corrections to Chapter~\ref{higherorder}, we concentrate here on the leading order which only involves the $T$-moduli $T^i=t^i+2i\chi^i$ and the singlet matter fields $C^L$. We will work in the usual approximation keeping only leading terms in $C^L$ and in inverse powers of the $T$-moduli. The usual K\"ahler potential for the $T$-moduli is given by
\begin{equation}
 \kappa_4^2 K_T=-\ln{{\cal V}}\; ,  \quad {\cal V}=\frac{1}{6}{\cal K}\; , \label{KT}
\end{equation}
where ${\cal V}$ is the Calabi-Yau volume and ${\cal K}$ is the cubic polynomial
\begin{equation}
 {\cal K}=d_{ijk}t^it^jt^k=\frac{1}{8}d_{ijk}(T^i+\bar{T}^i) (T^j+\bar{T}^j)(T^k+\bar{T}^k)\; .   \label{k3}
\end{equation}
The $C^L$ part of the K\"ahler potential has the form
\begin{eqnarray}\label{mKpot}
  K_{\rm matter} = G_{LM} C^L \bar{C}^{M}\; .
\end{eqnarray}
Here $G_{LM}$ is the matter field space metric, which depends on
the various moduli in the theory.  The precise form of this metric
will not be needed but it will be important that it is positive
definite.

\vspace{0.1cm}

We now have all of the information we require to compute the D-term
contribution to the four dimensional theory's potential.  Given that we have identified the
transformation properties of our low-energy fields, in particular under the $U(1)$ symmetry, 
this derivation is standard and can, in a somewhat different context, be found in the literature (see for example, \cite{Lukas:1999nh,Blumenhagen:2005ga}).
Nevertheless, we will carry this out explicitly, to present a complete and coherent argument.
According to the usual structure of four-dimensional ${\cal N}=1$ supergravity, the D-terms are determined by the following equations~\cite{Wess:1992cp}.
\begin{eqnarray} \label{deqns}
  g_{I \bar{J}} \bar{X}^{\bar{J} \eta} &=& i \frac{\partial}{\partial M^I} D^{\eta} \\
  g_{I \bar{J}} X^{I \eta}  &=& -i \frac{\partial}{\partial \bar{M}^{\bar{J}}} D^{\eta}
\end{eqnarray}
Here, the $M^I$ represent all of the fields in the theory, $g_{I \bar{J}}$
the complete field space metric, and $\eta$ is an index labeling the
adjoint of the gauge group. The quantities $X$ are the holomorphic
Killing vectors which generate those analytic isometries of the
K\"ahler field space which can be gauged. Under such a gauge transformation,
the fields $M^I$ then transform as
\begin{equation}
 \delta M^I=- \epsilon^\eta X^{I\eta}\; ,
\end{equation}
where $\epsilon^\eta$ are the gauge parameters. We can now determine the Killing vector for the $U(1)$ symmetry
by comparing this expression with the field transformations~\eqref{chiX} and \eqref{CX} which we have
derived from the higher-dimensional theory. This leads to
\begin{eqnarray}
  X^{i} &=&  i \frac{3}{8} \epsilon_S \epsilon_R^2 c_1^i({\cal F})\\
  X^{L}&=& i Q^{L} C^{L} \; .
 \end{eqnarray}
Inserting this Killing vector into \eqref{deqns} and solving for the associated $U(1)$ D-term we obtain
\bea \label{thedterm}
 D^{U(1)} =  \frac{3}{16} \frac{\epsilon_S \epsilon_R^2}{ \kappa_4^2}\frac{\mu ({\cal F})}{\cal V} - \sum_{L,\bar{M}}Q^{L} G_{L \bar{M}} C^L \bar{C}^{\bar{M}}\; ,
\eea
where $\kappa_4^2=\kappa_{11}^2/(v 2 \pi\rho)$ is the four-dimensional Planck constant.
Here, we have neglected contributions to this D-term from ${\bf 27}$ and $\overline{\bf 27}$ multiplets charged under the $U(1)$ symmetry. As long as $E_6$ remains unbroken these further contributions vanish and, for our explicit example, this will indeed be enforced by the $E_6$ D-terms. We note that the above D-term consists of a FI piece which is proportional to the slope
\begin{equation}
 \mu ({\cal F})=\frac{1}{2}c_1^i({\cal F})s_i=\frac{1}{2}d_{ijk}c_1^i({\cal F})t^jt^k=\frac{1}{8}d_{ijk}c_1^i({\cal F})(T^j+\bar{T}^j)(T^k+\bar{T}^k)\; ,
\end{equation}
of the destabilizing sub-sheaf ${\cal F}$ and a standard matter field piece. We also recall that ${\cal V}$ is the Calabi-Yau volume given in Eq.~\eqref{KT}.

\section{Stability Walls in the Effective Theory}\label{comparison}

In this section, we will study the vacuum structure of the effective
theory derived in the previous section. Our aim is to show how this
four-dimensional, field theory based analysis reproduces features seen
in the mathematical, ten-dimensional analysis of Section
\ref{bundlestab}. In other words, we would like to show how the
abstract mathematical concept of bundle stability and its implications
for supersymmetry can be understood in a physical way, from our
four-dimensional effective theory. It is clear from the
expression~\eqref{thedterm} for the D-term that the nature of the
four-dimensional vacuum space crucially depends on the charges $Q^L$
of the matter field singlets $C^L$. We begin with a general discussion
and then illustrate the main points with the example discussed in
Section \ref{simpleeg}.

We need to understand how the interplay between the FI and matter
field terms in Eq.~\eqref{thedterm} can reproduce the expected pattern
of broken or unbroken supersymmetry. A crucial observation is that the
FI term is proportional, with a positive constant of proportionality,
to the slope, $\mu ({\cal F})$, of the destabilizing sub-sheaf ${\cal
  F}$. We recall from our previous discussion that this slope is
negative in the part of the K\"ahler moduli space where the bundle is
stable and hence supersymmetric, and that it is positive where the bundle
breaks supersymmetry. The stability wall which separates these two
regions in K\"ahler moduli space is defined by $\mu ({\cal
  F})=0$. Given these features of the FI term, one can ask how the
D-term~\eqref{thedterm} for $\mu ({\cal F})<0$ can vanish and hence
preserve supersymmetry as we would expect. To achieve this, the FI term
obviously has to be cancelled by the matter field contribution in
\eqref{thedterm} through a suitable adjustment of the matter field
vevs. This will work precisely if there is {\em at least one
  negatively charged matter field} $C^L$, with $Q^L<0$ present. On the
other hand, if the D-term~\eqref{thedterm} is to be non-zero and thus
break supersymmetry for $\mu ({\cal F})>0$, as we expect it should,
{\em all matter fields need to be negatively charged}; that is, there
should be no matter fields with $Q^L>0$. The
D-term~\eqref{thedterm} then becomes a sum of two positive definite terms
in the non-supersymmetric region and there is no way in which they can
cancel each other.

Hence, for the D-term to correctly describe the expected pattern of
supersymmetry breaking the zero modes at the stability wall are
constrained in a specific way. Let us focus on our main class of
examples, namely bundles $V$ with $SU(3)$ structure group which
decompose as $V=\cF \oplus {\cal K}$, where $\cF$ is the rank two
de-stabilizing sub-sheaf. Then we can indeed show that the required
constraints on the particle spectrum are satisfied. We recall from
Table~\ref{table1} that the singlet matter fields $C^L$ correspond to
the cohomology groups $H^1(X, \cF \otimes {\cal K}^*)$ and
$H^1(X,\cF^* \otimes {\cal K})$, where the former leads to negative
and the latter to positive charge. Then one can show the following
\begin{lemma}\label{lemma_1}
Let $V$ be a holomorphic vector bundle with structure group $SU(3)$ defined over $X$, a Calabi-Yau $3$-fold. If $\cF$ is a rank $2$, stable sub-sheaf of $V$, defining the ``wall"  in the dual K\"ahler cone given by $\mu ({\cal F})=0$, such that $V$ is stable for $\mu(\cF)<0$ and unstable for $\mu(\cF)>0$, then $H^1(X, \cF \otimes (V/\cF)^*) \neq 0$ and $H^1(X, \cF^* \otimes V/\cF)=0$ (for any effective field theory describing only $V$).\footnote{If $H^1(X, \cF^* \otimes V/\cF)\neq 0$ then the bundle defined by the extension  $Ext^1(\cF,V/\cF)=H^1(X, \cF^* \otimes V/\cF)$ is {\it not} isomorphic to $V$. This case corresponds to a branch structure in the effective field theory which provides a transition to a new vector bundle and will be explored in more detail in \cite{Anderson}.}
\end{lemma}
\noindent The proof of this lemma (generalized to $SU(n)$ bundles) is provided
in Appendix A. It states that all singlet matter fields $C^L$ result
from the cohomology group $H^1(X, \cF \otimes {\cal K}^*)$ and
therefore, from Table~\ref{table1}, are  negatively charged,
as required. The fact that all of the fields $C^L$ carry $U(1)$
charges of the same sign means, of course, that the $U(1)$ symmetry is
anomalous. This is in line with expectations and we know that this
triangle anomaly is cancelled by the four-dimensional version of the
Green-Schwarz mechanism. Since we are dealing with a non-universal
anomalous $U(1)$, as discussed, this involves an anomalous variation
of the threshold correction to the gauge kinetic function induced by
the transformation~\eqref{chiX} of the $T$-axions. Details of this can
be found in Ref.~\cite{Lukas:1999nh}.

\vskip 0.4cm

We would now like to discuss the D-term~\eqref{thedterm} and its
associated vacuum space and particle masses in more detail. This will
provide us with a general picture of how the theory at the stability
wall relates to the standard heterotic low-energy theory at a generic
point in the supersymmetric part of K\"ahler moduli space. As
mentioned before, we will focus on the part of the moduli space where
$E_6$ is unbroken, so that we do not need to consider vevs of ${\bf
  27}$ and $\overline{\bf 27}$ multiplets. Hence, the fields of
central interest are the $T$-moduli $T^i=t^i+2i\chi^i$ and the singlet
matter fields $C^L$. It is clear that the D-term~\eqref{thedterm}
gives mass to precisely one real combination of the K\"ahler moduli
$t^i$ and the matter fields $C^L$, the Higgs field. Expanding
\eqref{thedterm} around a vacuum (that is, a vanishing D-term,
$D^{U(1)}=0$) by writing $t^i=\langle t^i\rangle +\delta t^i$ and
$C^L=\langle C^L\rangle +\delta C^L$, we find that this massive linear
combination is given by \bea \label{linearcombo} D^{U(1)} =
-\frac{3}{16} \frac{\epsilon_S \epsilon_R^2}{\kappa_4^2} G_{jk}
c_1^j({\cal F}) \delta t^k -\sum_{L,\bar{M}}Q^L G_{L \bar{M}} \left(
  \left< C^L\right> \delta \bar{C}^{\bar{M}} + \delta C^L \left<
    \bar{C}^{\bar{M}} \right> \right) \ , \eea where
\begin{equation}
 G_{ij}=-\frac{\partial^2\ln{\cal V}}{\partial t^i\partial t^j}\;  \label{G}
\end{equation}  
is the K\"ahler moduli space metric, expressed in terms of the
Calabi-Yau volume ${\cal{V}}$ as defined in Eq.~(\ref{KT}). In this
discussion, we are ignoring terms which are higher order in $\left< C^L
\right>$ and inverse powers of $t^i$.  The Goldstone mode, the
corresponding linear combination of $T$-axions $\chi^i$ and $C^L$
phases, is absorbed by the $U(1)$ vector boson in the super-Higgs
effect. Since supersymmetry is unbroken, the mass of the linear
combination~\eqref{linearcombo} and the $U(1)$ vector boson must be
equal and they can be computed from Eq.~\eqref{linearcombo} or from
the $\chi^i$ and $C^L$ kinetic terms. Either way one finds the mass is
given by
\begin{equation}
\label{u1mass} 
m_{U(1)}^2=\frac{1}{s}\left(\frac{(3 \epsilon_S \epsilon_R^2)^2}{256\kappa_4^2} c_1^i(\cF)c_1^j(\cF)G_{ij} +\sum_{L,\bar{M}}Q^LQ^{\bar{M}} G_{L \bar{M}} 
\langle C \rangle^L \langle \bar{C} \rangle^{\bar{M}}\right) \; ,
\end{equation}
where $s={\rm Re}(S)$ is the real part of the dilaton. To obtain this result from Eq.~\eqref{linearcombo} it is necessary to canonically normalise the kinetic terms $\frac{1}{4 \kappa_4^2} G_{ij}\partial\delta t^i\partial\delta t^j$ and $G_{L\bar{M}}\partial\delta C^L\partial\delta\bar{C}^{\bar{M}}$. 

Let us discuss this result, beginning at a point on the stability
wall. At the stability wall, $\mu(\cF)=0$ and it follows from
Eq.~\eqref{thedterm} that $\langle C^L\rangle =0$ in order to have a
vanishing D-term. Hence, at the stability wall the Higgs field is a
linear combination of K\"ahler moduli $\delta t^i$ only, while the
Goldstone mode consists of $T$-axions $\chi^i$. The $U(1)$ and Higgs
mass are then given by the first term in Eq.~\eqref{u1mass} which
scales like $1/(st^2)$ for a typical K\"ahler modulus $t$. This is to
be compared with the mass of a typical gauge sector massive mode
which scales as $1/(st)$. We see that the $U(1)$ and Higgs masses are
suppressed by a factor $1/t$ and, hence, that in the large radius
limit and close to the stability wall it is consistent to keep these
fields in the low energy theory.

What happens as we move away from the stability wall into the supersymmetric region? From Eq.~\eqref{thedterm} the matter field vevs $\langle C^L\rangle$ are now non-vanishing and their fluctuations $\delta C^L$ contribute to the Higgs fields and the Goldstone mode. Once we move sufficiently away from the stability wall, so that $\mu(\cF)=\cO (t^2)$,  Eq.~\eqref{thedterm} implies $G_{L\bar{M}}C^L\bar{C}^{\bar{M}}\sim 1/t$. Hence, far away from the stability wall, the $U(1)$ mass~\eqref{u1mass} scales as $1/(st)$ and becomes comparable to a typical heavy gauge sector mass. In this limit, we should, therefore, remove the $U(1)$ vector multiplet and the Higgs multiplet from the low-energy theory. In this way, we recover the standard $E_6$ gauge group at a generic supersymmetric point in the K\"ahler moduli space. 

What about the matching of chiral multiplets to the usual analysis?
First, we note that far away from the stability wall the Higgs
multiplet becomes predominantly a linear combination of the matter
fields $C^L$. This means that there are massless ``$T$-moduli'' in
this region, which are slightly corrected versions of the naively
defined fields, consistent with the expectation from standard
heterotic compactifications. As for the $E_6$ singlet matter fields,
at the stability wall we have $h^1(X,\cF\otimes (V/\cF)^*)$ fields
$C^L$ and $h^1(X,\cF\otimes\cF^*)$ bundle moduli. Away from the
stability wall, one combination of $C^L$ fields is removed from the
low-energy theory so that we remain with $h^1(X,\cF\otimes
(V/\cF)^*)+h^1(X,\cF\otimes\cF^*)-1$ singlet fields. To match standard
heterotic compactifications, this must equal $h^1(X,V\otimes V^*)$,
the number of bundle moduli at a generic supersymmetric point in
K\"ahler moduli space. That is is indeed always the case is stated in
the following lemma.
\begin{lemma}\label{lemma_2}
Let $V$ be a holomorphic vector bundle with structure group $SU(3)$ defined over $X$, a Calabi-Yau $3$-fold. If $\cF$ is a rank $2$, stable sub-sheaf of $V$, defining the ``wall" in the dual K\"ahler cone given by $c_1^i(\cF) s_{i}=0$, such that $V$ is stable for $\mu(\cF)<0$ and unstable for $\mu(\cF)>0$, and further, $H^1(X, \cF^* \otimes V/\cF)=0$, then 
\beq h^1(X,V \otimes V^*)=h^1(X, \cF \otimes (V/\cF)^*)+h^1(X, \cF \otimes\cF^* )-1~,\eeq 
where $h^1(X,V\otimes V^*)$ is the generic dimension of bundle moduli space when $V$ is a stable bundle.\end{lemma}
\noindent The proof of this lemma can be found in Appendix A.

In summary, we see that the D-term~\eqref{thedterm} correctly
reproduces all of the expected physical features of gauge bundle
supersymmetry. Specifically, the D-term vanishes and, hence, preserves
supersymmetry precisely in the region where the gauge bundle is stable
while it is non-vanishing in the region where the bundle is
unstable. We have seen that in the large radius limit and at the
stability wall it is consistent to keep the massive $U(1)$ vector
multiplet and the Higgs multiplet in the low energy theory. Away from
the stability wall, however, these fields develop heavy masses and
have to be dropped. In this way, we recover the usual heterotic
effective theory at a generic, supersymmetric point in moduli
space. 

It is clear that the physics of the non-supersymmetric region of the
K\"ahler cone is dominated by the potential wall due to the
non-vanishing D-term~\eqref{thedterm}. Since there is no perturbative
vacuum in this region, we shall refrain from discussing the mass
spectrum in this part of field space. However, to finish this section
we shall make a few comments about the regime of the
non-supersymmetric region where our effective field theory analysis is
valid. In addition to the usual expansions of heterotic M-theory,
which will be discussed in more detail in section \ref{higherorder},
validity of our approach requires that the potentials present should
be below the compactification scale. Furthermore, the $C^L$ field vevs
should not be too large, since we have assumed they were small in
deriving the effective potential in this section. Obviously, both of
these conditions are satisfied close to the transition between the
supersymmetric and non-supersymmetric regions of moduli space and, as
such, the above discussion can be trusted.  Far into the
non-supersymmetric region one may not expect a four dimensional
description to exist at all. The potential grows in size as we
penetrate inside this zone until eventually it becomes of the same
mass scale as heavy states which have been truncated in our
analysis. To give some idea of scale, let us examine the size of the
potential in the non-supersymmetric region when all $C^{L}$ vevs
vanish. At a typcial non-supersymmetric point in field space, the
ratio of this potential to the fourth power of a typical mass of a
heavy gauge sector state is of order $s$, the dilaton, when working in
string units. As such, in a valid regime of the effective theory where
$s$ is large, one should typically not include regions with such a
potential in the four-dimensional theory. Close to the boundary with
the supersymmetric region (where the D-term potential vanishes
exactly), however, the potential is surpressed from its usual scale by
the smallness of $\mu({\cal F})^2 / {\cal V}^{4/3}$, which smoothly
increases from zero as we enter the non-supersymmetric part of the
K\"ahler cone. Thus, we can trust our analysis and investigate this
potential in the four-dimensional theory in the region near to the
boundary where $(\mu({\cal F})^2/{\cal V}^{4/3}) s <<1$.

\subsection{An example}\label{thefirstexample}
To illustrate the above general discussion, let us return to the
example of Section \ref{simpleeg}.  Recall, that we have defined the
monad bundle, \eref{monadeg} on the complete intersection Calabi-Yau
manifold~\eqref{cicy24}.  As mentioned in Section \ref{simpleeg}, we
find that the $SU(3)$ bundle, $V$, decomposes as $V \rightarrow \cF
\oplus \mathcal{K}$ where $\mathcal{K}=\cO_X(1,-1)$ is a line
bundle. The de-stabilizing sub-sheaf $\cF\subset V$ has rank
two\footnote{In this example $\cF$ is a bundle that injects into $V$
  everywhere in moduli space, while $\mathcal{K}$ is a line-bundle and
  \textit{only} injects at the decomposable point. That $\cF$ is
  indeed a bundle and not simply a sheaf has been checked explicitly
  using the computer algebra packages~\cite{Gray:2008zs}.}  and is
described by the monad 
\beq 0 \to \cF \to \cO_{X}(1,0) \oplus
\cO_{X}(0,1)^{\oplus 2} \to \cO_{X}(2,1) \to 0\; . \label{Fmonad}
\eeq 
The locus in the moduli space of $V$ where it decomposes as $V=\cF \oplus
\mathcal{K}$ corresponds to setting to setting to zero the bi-degree
$(1,2)$ polynomials in the monad map, $f$, given in
\eref{monadeg}. Using the results of
Refs.~\cite{Anderson:2008uw,stability_paper}, we can calculate the
dimensions of the cohomology groups of $\cF$ and $\mathcal{K}$ listed
in Table \ref{table1}. The results are summarised in
Table~\ref{table2}.
\begin{table}
\begin{center}
\begin{tabular}{|c|c|c|c|}
  \hline
  Representation &Cohomology &Physical $U(1)$charge &Dimension of Cohomology\\ \hline
  $({\bf 1},{\bf 2})_{-3}$ &  $H^1(X, {\cal F} \otimes {\cal K}^*)$ & $-3/2$ & $16$ \\ \hline
  $({\bf 1},{\bf 2})_{3}$ &  $H^1(X, {\cal F}^* \otimes {\cal K})$ & $3/2$ & $0$ \\ \hline
  $({\bf 1},{\bf 3})_0$ &  $H^1(X, {\cal F} \otimes {\cal F}^*)$ & $0$ & $7$ \\ \hline
  $({\bf 27},{\bf 1})_{2}$ & $H^1(X, {\cal K})$ & $1$ & $0$ \\ \hline
  $({\bf 27},{\bf 2})_{-1}$ & $H^1(X, {\cal F})$ & $-1/2$ & $2$ \\ \hline
  $(\overline{\bf 27},{\bf 1})_{-2}$ & $H^1(X, {\cal K}^*)$ & $-1$ & $0$ \\ \hline
  $(\overline{\bf 27},{\bf 2})_{1}$ & $H^1(X, {\cal F}^*)$ & $1/2$ & $0$\\ \hline
\end{tabular}
\mycaption{Particle content of the model defined by the bundle~\eqref{monadeg} at the decomposable locus where
$V={\cal F}\oplus{\cal K}$, with ${\cal F}$ defined by \eqref{Fmonad} and ${\cal K}=\cO_X(1,-1)$.}
\label{table2}
\end{center}
\end{table}
The {\it only} matter fields present which are charged under the
additional $U(1)$ symmetry appear in the first and fifth row in the
table. They both have negative charge under the four-dimensional
$U(1)$, listed in the third column.  In particular, this means that
the singlet matter fields $C^L$, which correspond to the first row in
the table, are all negatively charged, in accordance with Lemma
1. Further, in this particular model it turns out that the ${\bf 27}$
matter multiplets are also negatively charged.  This means, by gauge
invariance, that the F-term part of the potential vanishes. Having
only ${\bf 27}$ but no $\overline{\bf 27}$ multiplets means the ${\bf
  27}$ vevs will be forced to zero by the $E_6$ D-terms. Hence, they
do not contribute to the $U(1)$ D-term~\eqref{thedterm}. For the
present example and all models with similar particle content, the
$U(1)$ D-term~\eqref{thedterm} therefore describes the full vacuum
space. In general, models with positively charged $E_6$ multiplets or
anti-families in $\overline{\bf 27}$ exist. For such models one would
expect superpotential terms or D-flat directions with non-vanishing
${\bf 27}$ and $\overline{\bf 27}$ vevs, leading to a more complicated
structure of the vacuum space. As explained before, for such models
the $U(1)$ D-term~\eqref{thedterm} describes the part of the vacuum
space where $E_6$ is unbroken. A generalised expression, including the
family and anti-family degrees of freedom in the D-term, can trivially
be derived.

From Eqs.~\eqref{thedterm} the D-term for this example reads
\begin{equation}
\label{dtermex}
D^{U(1)} = \frac{3}{16} \frac{\epsilon_S \epsilon_R^2}{\kappa_4^2}\frac{\mu ({\cal F})}{\cal V} +\frac{3}{2} \sum_{L,\bar{M}=1}^{16}G_{L \bar{M}} C^L \bar{C}^{\bar{M}}\; ,
\end{equation}
where, from Eqs.~\eqref{slopeex}, \eqref{sdefex}, the slope is given by
\begin{equation}
 \mu(\cF)=\frac{1}{2}(-s_1+s_2)\; ,\quad s_1=4(t^2)^2\; ,\quad s_2=8t^1t^2+2(t^2)^2\; .
\end{equation} 
For the volume we have
\begin{equation}
 {\cal V}=2t^1(t^2)^2+\frac{1}{3}(t^2)^3\; .
\end{equation} 
In Figure \ref{plot1} we plot the D-term potential~\eqref{dtermex} as a function of the
dual K\"ahler cone variables $s_1$ and $s_2$, defined in Eq.~\eqref{sdefex}. The $C^L$ vevs, which are not plotted due to lack of dimensions, have been chosen to be at their minimum. The potential
rising from zero in the unstable region is clearly visible, as is the stability wall determined by the line with slope $1$ (in agreement with the bundle stability regions shown in Figure \ref{f:example7887}).
\begin{figure}[!ht]
\centerline{\epsfxsize=4in\epsfbox{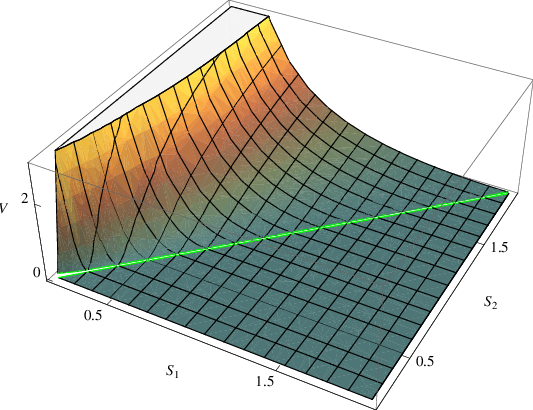}}
\mycaption{The potential in the dual K\"ahler cone as a functions
    of the two dual K\"ahler variables. The potential has been
    minimized with respect to the $C^{L}$ fields (which are not plotted here). The
    flat region of the potential is where the bundle is
    stable. The positive definite potential wall which one encounters
    upon entering the region where the bundle is unstable can clearly
    be seen, arising at the line with slope $=1$.}
\label{plot1}
\end{figure}
Figure~\ref{plot2} shows the D-term potential from Eq.~\eqref{dtermex}
as a function of the coordinate $s_1$, with $s_2$ chosen such that the
plot traces a line perpendicular to the stability wall in the K\"ahler
cone, and the radius $|C|$ of a representative singlet matter
field. This figure makes it clear that there is no ``boundary'' to the
vacuum space at the stability wall in K\"ahler moduli space if one
considers the full field space of the theory.
\begin{figure}[!ht]
  \centerline{\epsfxsize=4in\epsfbox{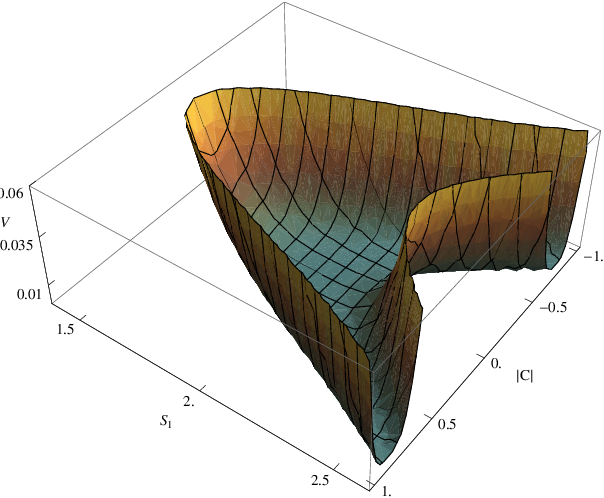}} \mycaption{ The
    D-term potential from Eq.~\eqref{dtermex}, as a function of $s_1$,
    a dual K\"ahler modulus, and the absolute value, $|C|$, of a
    representative singlet matter field $C$. In this plot we have
    chosen $s_2=4-s_1$ so that we are examining a line in K\"ahler
    moduli space perpendicular to the boundary between the
    supersymmetric and non-supersymmetric regions. The boundary itself
    is found at $s_1=2$ in this diagram. Since the exact form of the
    K\"ahler potential for the matter fields is not known, a simple,
    canonical form has been chosen for illustrative
    purposes.\label{plot2}}
\end{figure}

At the stability wall, where $t^2=4t^1$ and $\langle C^L\rangle =0$, the variation of the D-term~\eqref{dtermex} becomes
\begin{equation}
 D^{U(1)}=\frac{9 \epsilon_S \epsilon_R^2}{640\kappa_4^2}\frac{1}{(t^1)^2}(4\delta t^1-\delta t^2)\; .
\end{equation}
This shows that it is the combination $4\delta t^1-\delta t^2$ of K\"ahler moduli perpendicular to the stability wall which becomes massive at this point, as is also evident from Figure~\ref{plot2}. The mass of this linear combination is given by
\begin{equation}
 m_{U(1)}^2=\frac{3 (\epsilon_S \epsilon_R^2)^2}{256 \kappa_4^2}\frac{1}{s(t^1)^2}\; .
\end{equation} 
This expression shows explicitly the aforementioned $1/t^2$ scaling of
the $U(1)$ vector and Higgs masses which justifies keeping these
states in the low-energy theory close to the stability wall. As
discussed earlier, far away from the stability wall the Higgs
multiplet becomes pre-dominantly a linear combination of the $C^L$
multiplets and there are two massless K\"ahler moduli as one would
expect at a generic point in the supersymmetric region. From
Table~\ref{table2}, we have $16$ singlet matter fields $C^L$ and $7$
bundle moduli at the stability wall. With one of the $C^L$ becoming
massive one would expect $16+7-1=22$ bundle moduli at a generic
supersymmetric point in moduli space and this is indeed the number we
have computed for this example, see Eq.~\eqref{bundlemod}. This
illustrates the general statement in Lemma 2. An additional example, with an unrelated manifold and method of bundle construction is provided in Appendix B.

\subsection{Results in bundle stability from the effective theory}

So far, we have used mathematical information on vector bundle
stability to construct a low-energy description of bundle
supersymmetry. Now that we have established such a picture, let us
reverse our approach and see if can recover some of the mathematical
results in bundle stability from the effective field theory. A key
fact to remember is the interpretation of the charged matter fields,
$C^L$, as bundle moduli of the $SU(3)$ bundle - as described in
Section \ref{eft}. When these fields vanish the bundle decomposes as
$V = {\cal F} \oplus {\cal K}$, and has structure group $S(U(2) \times
U(1))$. For non-vanishing $C^L$ vevs, at a generic point in moduli
space, the bundle no longer splits up into a direct sum of
sub-bundles, and the structure group reverts to $SU(3)$.

First, in what we would expect to be the supersymmetric region of
K\"ahler moduli space, the fields $C^L$ {\it must} acquire a vev if
the D-term is to vanish. From this observation, we reproduce the fact
that the bundle will only produce a supersymmetric vacuum in the
so-called ``stable" region of K\"ahler moduli space, if it is at a
{\it generic} (that is, non-split) point in its moduli space, that is,
if the structure group is $SU(3)$.  If, in the normally ``stable"
region of K\"ahler moduli space, the bundle moduli move to the
decomposable locus where the structure group is $S(U(2)\times U(1))$,
the D-term is non-vanishing and supersymmetry is broken. This is all
in perfect agreement with the algebro-geometric analysis presented in
section \ref{stability}.

As we learned in Section \ref{stability}, at the stability wall in
K\"ahler moduli space, in order to have a supersymmetric theory, the
bundle must be split and semi-stable; that is, it must decompose into a direct
sum of stable bundles of the same slope. From \eqref{thedterm} we
again see this behaviour reproduced. The FI term vanishes on this
line in K\"ahler moduli space. Hence, the vanishing of the D-term
required by supersymmetry forces the $C^L$ field vevs to vanish -
taking us precisely to the split point in bundle moduli space.

As before, our discussion here allows us to go further than has been
previously possible and discuss what happens in the region where
supersymmetry is spontaneously broken as well. Although the
D-term~\eqref{thedterm} cannot vanish in this part of moduli space,
for fixed K\"ahler moduli, the D-term potential can be minimized by
vanishing fields $C^L$. Thus, the bundle will relax to the
decomposable locus in bundle moduli space throughout this region, as
well as at the stability wall, in the absence of non-perturbative
effects.

As a final comment, it is interesting to note that the D-term \eqref{thedterm} does not
depend on the complex structure fields. Thus, it should also be true
that the stability regions derived in Section \ref{bundlestab} are not
dependent on the choice of complex structure, for those bundles which
can give rise to supersymmetric theories in four dimensions. This
result, which is somewhat surprising from a mathematical perspective,
will be discussed further in the Appendix. 

\section{Higher Order Corrections} \label{higherorder}

In the analysis of proceeding sections we have worked to first order
in the strong coupling expansion parameter, $\epsilon_S$, and the
square of the matter fields. The strong coupling expansion parameter
$\epsilon_S$ itself, as opposed to the combination $\epsilon_S
\epsilon_R^2$ which is what was defined in \eqref{combodef} and has
appeared heretofore, is given by the following \cite{Witten:1996mz, Lukas:1998hk},
\begin{equation}
 \epsilon_S=\left(\frac{\kappa_{11}}{4\pi}\right)^{2/3}\frac{2\pi\rho}{v^{2/3}}\; .
\end{equation}
In the weakly coupled langauge, we have been working, up to this point,
at string theory tree-level. One can do better than this and work out
those $\cO(\epsilon_S^2)$ corrections that correspond to string
one-loop corrections\footnote{Corrections corresponding to higher orders in $\alpha'$ would require knowledge of the K\"ahler potential for bundle moduli, which is only known for special cases \cite{Gray:2003vw}. For a discussion of higher order
  corrections in $\alpha'$ to the supersymmetry/slope stability
  condition in the Type II context, see e.g. \cite{Douglas:2000ah}.}.
In particular, the D-term given in \eqref{thedterm}, receives
$\epsilon_S^2$ corrections which can be calculated. Corrections to the
matter field part of the D-term~\eqref{thedterm} are
uninteresting. The only fact that we have used about this term is the
positive definite nature of the matter field metric $G_{L \bar{M}}$,
and this will not be changed by such corrections. However, the
$\cO(\epsilon_S^2)$ corrections to the FI term are of some interest
and we now proceed to derive these.  At lowest order, the $T$-moduli
had a non-trivial $U(1)$ transformation while all other moduli fields
were invariant. As we will see, at higher order, the dilaton $S$ and
the five-brane position moduli $Z^\alpha$, where $\alpha =1,\dots ,N$
numbers the different five-branes, also transform non-trivially. We
start by defining these four-dimensional superfields in terms of the
underlying geometric fields. The definition of the $T$-moduli,
$T^i=t^i+2i\chi^i$, is as previously given (see Eq.~\eqref{Tdef}). For
the dilaton and the five-brane moduli we
have~\cite{Brandle:2001ts,Mattias} \bea S &=& V_0 +
\pi\epsilon_S\sum_{\alpha=1}^N \beta^\alpha_i t^i z_\alpha^2 +
i\left(\sigma
  +2\pi\epsilon_S \sum_\alpha^N \beta^\alpha_i \chi^i  z_\alpha^2\right)\label{Sdef}\\
Z^\alpha &=& \beta^\alpha_i \left( t^i z_\alpha + 2 i
  (-n_\alpha^i\nu_\alpha + \chi^i z_\alpha)\right)\; .\label{Zdef}
\eea Here, $V_0$ is the Calabi-Yau volume averaged over the orbifold
and $\sigma$ is the dilatonic axion, the dual of the four-dimensional
two-form $B_{\mu\nu}=C_{11\mu\nu}$. Further, $z_\alpha$ is the
distance from the left orbifold fixed plane to the $\alpha$-th
five-brane, the $\beta^\alpha_i$ are the charges associated to the
$\alpha$-th five-brane and
$n_\alpha^i=\beta^\alpha_i/(\sum_i(\beta^\alpha_i)^2)$. The fields
$\nu_\alpha$ are axions located on the five-brane world-volumes.  In
order to compute the corrections to the D-term, we need to consider the
$U(1)$ transformations of the fields at order $\epsilon_S^2$. For the
$T$-moduli and the matter fields, these transformations are given in
Eqs.~\eqref{chiX} and \eqref{CX} with no further corrections at
$\cO(\epsilon_S^2)$.  The transformation of the dilaton and five-brane
position superfields are slightly more subtle in their origin. To
discuss the dilaton, we consider the relevant terms in the
four-dimensional effective action which involve the two-form
$B_{\mu\nu}=C_{11\mu\nu}$. These terms are~\cite{Lukas:1999nh} \bea
S_{4d,B} =-\frac{1}{2 \kappa_4^2} \int_{{\cal M}_4} \left[ V_0^2 H
  \wedge *H + \frac{3}{4} \pi \epsilon_S^2 \epsilon_R^2 c_1^i({\cal
    F}) \beta_i B \wedge F \right] \; , \eea where $H=dB+\dots$ and
the dots indicate a Chern-Simons three-form which is irrelevant for
the present discussion. Further, $F=dA$ is the field strength of the
$U(1)$ gauge field $A$ and the integer charges $\beta_i$ of the $E_8$
sector under consideration are defined as \bea \beta_i = \frac{1}{16
  \pi«^2} \int_X \left( \textnormal{tr} F\wedge F - \frac{1}{2}
  \textnormal{tr}R\wedge R\right)\wedge J_i \; .  \eea In order to
dualise the two-form $B$ to the dilatonic axion $\sigma$, we set
$H_0=dB$ and add to the above action the term \bea
\frac{1}{\kappa_4^2} \int_{{\cal M}_4} H_0\wedge d \sigma\;.  \eea By
integrating out $H_0$, we find the kinetic term \bea
S_{4d,\textnormal{dual}} =-\frac{1}{\kappa_4^2} \int_{{\cal M}_4}
\left( \frac{1}{V_0^2} \Sigma \wedge * \Sigma \right)   \eea for
the dilatonic axion $\sigma$, where the ``field strength'' $\Sigma$ is
defined as \bea \Sigma= d \sigma - \frac{3}{8} \pi \epsilon_S^2
\epsilon_R^2 c_1^i({\cal F}) \beta_iA\; .  \eea This field strength
needs to be invariant under $U(1)$ gauge transformations with $\delta
A=-D\tilde{\epsilon}$, which implies the following transfomation law
for the dilatonic axion.  \bea \label{sigmaX} \delta \sigma =
-\frac{3}{8} \pi \epsilon_S^2 \epsilon_R^2 c_1^i({\cal F}) \beta_i \;
\tilde{\epsilon}\; .  \eea The five-brane axions $\nu_\alpha$ do not
transform under $U(1)$ transformations, so the $\chi^i$
transformation~\eqref{chiX} and the above $\sigma$
transformation~\eqref{sigmaX} are all we have to take into account at
the component field level. Note that, from Eqs.~\eqref{Sdef},
\eqref{Zdef} and \eqref{Tdef}, this implies non-trivial
transformations for all superfields $S$, $Z^\alpha$ and $T^i$. In
particular, the five-brane moduli superfields $Z^\alpha$ in
Eq.~\eqref{Zdef} pick up a non-trivial transformation through their
dependence on the $T$-axions $\chi^i$.

Taking these new field transformations into account, we may now
calculate the correction to our D-term, \eqref{thedterm}, at order
$\epsilon_S^2$.  For this we need the relevant corrections to the
K\"ahler potential. In Eq.~\eqref{KT} we have already given the
K\"ahler potential for the $T$-moduli which remains unchanged at the
orders we require. The K\"ahler potential for the dilaton and the
five-brane moduli is given by
\begin{equation}
K_S=-{\rm ln}\left[S+\bar{S}- \pi\epsilon_S \sum_{\alpha=1}^{N}\frac{(Z^\alpha+\bar{Z}^\alpha)^{2}}
    {\beta_i^\alpha(T^i+\bar{T}^i)}\right]\
\end{equation}
Given these expressions, we may follow exactly the same procedure as in
Section \ref{eft} to obtain the corrected D-term
\bea \label{correcteddterm} D^{U(1)} = f - \sum_{L \bar{M}} Q^L G_{L \bar{M}}
C^L\bar{C}^{\bar{M}}\; .  \eea Here, the FI term $f$ is given by \bea
f &=& f^{(0)} + f^{(1)} \\
f^{(0)}&=& \frac{3}{16} \frac{\epsilon_S \epsilon_R^2}{\kappa_4^2}\frac{\mu(\cF)}{\cV}\\
f^{(1)}&=& \frac{3 \pi \epsilon_S^2 \epsilon_R^2}{8
  \kappa_4^2}\frac{1}{S+\bar{S}}\left[\beta_ic_1^i(\cF) +
  \pi\sum_{\alpha=1}^N\frac{(Z^\alpha+\bar{Z}^\alpha)^2}{(\beta_i^\alpha
    (T^i +\bar{T}^i))^2}\beta_i^\alpha c_1^i({\cal F})\right]
\label{2o}
\eea We see that the leading contribution, $f^{(0)}$, to the FI term
precisely reproduces our previous result~\eqref{thedterm} while the
correction term $f^{(1)}$ is surpressed by an extra power of
$\epsilon_S$, as expected. As mentioned earlier, the second term in
\eqref{correcteddterm} will also receive corrections. However, since
these small corrections cannot change the sign of this term they are
of no immediate interest to us here.

The $\cO(\epsilon_S^2)$ correction $f^{(1)}$ to the FI term depends on fields other
than the K\"ahler moduli. This means that the position of the stability wall in the K\"ahler cone will change
slightly as we change, for example, the value of the dilaton or the five-brane moduli $Z^\alpha$. Naively, this suggests that
we have lost the link, as espoused in the rest of the paper, between the mathematical stability analysis and the four-dimensional effective field theory. However, this is not the case.

The crucial point is that the four-dimensional fields which appear in
the above expression are not quite those which are ``experienced
by the gauge fields''. In heterotic M-theory, the vacuum solution in
eleven dimensions includes a warping in the eleventh direction which
introduces dependence of the K\"ahler moduli on the orbifold coordinate.
In other words, the six dimensional manifold changes shape slightly as we
traverse the $S^1/\mathbb{Z}_2$ orbifold direction. The four dimensional K\"ahler moduli $t^i$ which appear in the above
expressions (for example, in Eq.~\eref{correcteddterm}) are the orbifold {\it average} of these varying K\"ahler parameters. The gauge fields of our bundle, however, reside on one of the orbifold fixed planes at
either end of the interval. Thus, in performing the stability analysis
of Sections \ref{stability} and \ref{bundlestab}, it is not the
averaged quantities which are relevant, but the K\"ahler moduli of the
Calabi-Yau $3$-fold at the relevant orbifold fixed plane. It is precisely the difference between those K\"ahler parameters at the orbifold fixed plane and the averaged ones which accounts for the
correction given in equation \eqref{2o}. This may be checked
explicitly using the expressions for the warping of heterotic M-theory
given in Refs.~\cite{Lukas:1997fg,Lukas:1998yy,Lukas:1998tt,Mattias}. Note that, in the case of
Abelian bundles, such corrections have been discovered elsewhere in
the literature~\cite{Blumenhagen:2005ga,Blumenhagen:2006ux}.

To make this precise, let us drop the requirement that we write the FI
term in terms of four-dimensional superfields. Instead, we introduce
the K\"ahler moduli $\tilde{s}_i$ of the Calabi-Yau manifold on the
relevant orbifold fixed plane (as opposed to the averaged K\"ahler
moduli $s_i$) and denote by $\tilde{\mu}(\cF)=c_1^i(\cF)\tilde{s}_i/2$
and $\tilde{\cV}$ the corresponding slope and volume. Then one can
show that the corrected D-term~\eqref{correcteddterm} can be written
as \bea D^{U(1)} = \frac{3}{16} \frac{\epsilon_S
  \epsilon_R^2}{\kappa_4^2}\frac{\tilde{\mu}(\cF)}{\tilde{\cV}}
-\sum_{L,\bar{M}} Q^L G_{L \bar{M}} C^L \bar{C}^{\bar{M}} \eea All
correction terms have disappeared and the FI term is proportional to
the slope computed for K\"ahler parameters on the orbifold plane,
where the bundle is actually defined. This is precisely the slope one
would define in a mathematical context. Hence, our interpretation of
the $U(1)$ D-term in terms of gauge bundle stability is completely
unchanged by higher order corrections.

\section{Conclusions and Further Work} \label{conclusions}

In this paper, we have explored in detail the structure of heterotic
theories near a stability wall, separating regions in K\"ahler moduli
space where a non-Abelian internal gauge bundle preserves or breaks
supersymmetry. We have found four-dimensional effective theories valid
near such boundaries which provide us with an explicit low-energy
description of bundle supersymmetry breaking and with a physical
picture for the mathematical notion of slope stability.  A key
observation in our analysis is that at a stability wall the structure
group of the internal gauge bundle decomposes and acquires a $U(1)$
factor. This leads to an additional $U(1)$ symmetry in the
four-dimensional effective theory which is Green-Schwarz
anomalous. The associated $U(1)$ D-term consists of a FI term and a
matter field term and it controls the supersymmetry properties of the
bundle from a four-dimensional point of view. Specifically, the FI
term is proportional to the slope $\mu(\cF)$ of the destabilizing
sub-sheaf $\cF\subset V$ of the internal vector bundle $V$. For
negative slope the bundle $V$ is stable. In the four-dimensional
theory this is reproduced, since non-trivial vacuum expectation values of
$U(1)$ charged matter fields compensate the FI term so that the $U(1)$
D-term vanishes and supersymmetry is preserved. For positive slope,
that is an unstable bundle $V$, the FI term changes sign. As all
$U(1)$ charges have the same sign, the FI term cannot be cancelled by
matter field vevs in this case and supersymmetry is broken. In four
dimensions, the relation between the theory at the stability wall and
at a generic supersymmetric point is governed by the super-Higgs
effect. As one moves away from the stability wall the $U(1)$ vector
field mass increases and has to be removed from the low-energy theory,
together with the associated Higgs multiplet.  The implied matching of
degrees of freedom can be precisely reproduced by a cohomology
calculation.  We have also shown that our results are robust under
corrections suppressed from the leading effects by a power of
$\epsilon_S$ (the strong coupling expansion parameter), corresponding
to string one-loop corrections. While the FI term does receive
corrections at this order, they have a simple interpretation in terms
of 11-dimensional geometry. While the standard four-dimensional
K\"ahler moduli $t^i={\rm Re}(T^i)$ measure the {\it average}
Calabi-Yau size across the orbifold, the gauge bundle and its
stability properties are sensitive to the Calabi-Yau moduli,
$\tilde{t}^i$, on the relevant orbifold fixed plane. The order
$\epsilon_S^2$ corrections to the FI term simply accounts for the
difference between those two types of moduli when the D-term is
expressed in terms of the standard four-dimensional fields $t^i$. In
other words, the order $\epsilon_S^2$ terms disappear when the D-term is
written in terms of $\tilde{t}^i$. Hence, these one-loop corrections do
{\it not} suggest a modification of the mathematical notion of bundle
stability but simply reflect the fact that the gauge fields are
localised in the orbifold direction.
We stress that the basic picture we provide here, while illustrated
for the sake of clarity with vector bundles with $SU(3)$ structure
group decomposing into $S(U(2)\times U(1))$, is very general. We expect
its main features to holds for any Calabi-Yau three-fold and for any
construction of vector bundles. Indeed, the validity of our approach
has been checked in a large number of disparate examples.

\vskip 0.4cm

Our results suggest many further directions for research, some
mathematical in nature and some physical.  It would be of great
interest to study various generalisations and extensions of the
mechanism described in this paper. In the present paper, we have
focused, when describing examples, on simple cases with two K\"ahler
moduli, so that the stability walls in K\"ahler moduli space are
lines. We stress, however, that the phenomenon we have described is
much more general and appears in K\"ahler cones of any dimensionality
greater than one. In general, the stable region is a sub-cone of the
K\"ahler cone with each co-dimension one face giving rise to a D-term
of the type we have described. At each generic point on the stability
wall only one of these D-term will be relevant. However, for more than
two K\"ahler moduli co-dimension one faces can intersect so that there
are special loci on the stability wall where two or more D-terms need
to be considered at a time. Further study of more complicated examples
would be an interesting future line of research.  Further
generalisation could involve considering more complicated splitting
types at the stability wall, such as $SU(3)\rightarrow S(U(1)\times
U(1)\times U(1))$, and $SU(n)$ bundle structure groups with
$n>3$. Indeed, the authors have already studied such cases in detail
and hope to present examples of this type in future work. An
interesting observation is that the four-dimensional effective field
theory only depends on the structure of the gauge bundle at the split
locus in moduli space. This suggests that phenomena similar to the
ones described here can link nominally different bundles together via
smooth transitions in physical moduli space. The authors are currenty
actively investigating this effect.

From a more phenomenological perspective, the potential we provide may
be of some interest in moduli stabilization \cite{burt2}. Its
perturbative nature means that this potential is relatively
steep. Thus, if one were to balance it against a non-perturbative
potential, such as that due to membrane instantons, one might be able
to obtain a naturally small scale of supersymmetry breaking. An
investigation of whether such an idea is phenomenologically viable is
underway. Global remnants of the anomalous $U(1)$ symmetry at the stability wall may have implications for the structure of the theory even at a generic supersymmetric point in moduli space. For example, one might be able to conclude that certain superpotential terms are forbidden. Such considerations may be used to constrain the type of vector bundles which can lead to realistic low-energy models.

Finally one can imagine attempting to use the analysis described in
this work to investigate what may be said about bundle stability
purely from the point of view of the four-dimensional effective
theory. One goal of such work would be to give a simple set of rules, for example based
on four-dimensional anomaly cancellation, which would guarantee that a given vector bundle on a Calabi-Yau manifold
is stable in a certain region of moduli space.

\section*{Acknowledgments}
The authors would like to thank to Nathan Seiberg, Juan Maldacena, Ron Donagi, and Ignatios Antoniadis for useful discussions. The work of L.~A.~ and B.~A.~O. is supported in part by the DOE under contract No. DE-AC02-76-ER-03071. Further, B.~A.~O. would like to acknowledge the Ambrose Monell Foundation at the IAS for partial support. A.~L. is supported by the EC 6th Framework Programme MRTN-CT-20040503369. J.~G. is supported by STFC UK.

\section{Appendix A: Two lemmas and a conjecture}\label{Appendix}

In this section, we will state the two lemmas used in Section \ref{comparison} (regarding the dimensions of certain cohomology groups) somewhat  more formally and provide proofs. These results will be an example of the types of cohomology conditions one can derive in the context of slope stability. Similar conditions can be derived when different $SU(n)$ bundle decompositions are considered or when additional enhanced $U(1)$ symmetries are present. Furthermore,  we will make a conjecture regarding the complex structure dependence of a stability wall.

Let $X$ be a Calabi-Yau three-fold with K\"ahler form $J$ and $V$ a holomorphic vector bundle
defined over $X$ with structure group $SU(n)$, where $n=3,4,5$. We will
consider a case in which a single sub-sheaf ${\cal F}\subset V$ of rank $n-1$
de-stabilizes $V$ in some part of the K\"ahler moduli space of $X$. We define the slope, $\mu (\cF)$, of
$\cF$ for a given polarization $J=t^k J_k$ by
\begin{equation}
 \mu(\cF)=\frac{1}{{\rm rk}(\cF)}\int_Xc_1(\cF)\wedge J\wedge J\; .
\end{equation} 
Let us further suppose that $\cF$ itself is slope-stable and has a slope such that
it destabilizes only part of the K\"ahler cone (as in Fig.~\ref{f:stab} in Section \ref{bundlestab}). Thus, $V$ is stable for
polarizations $J$ with $\mu(\cF)<0$ and unstable for polarizations $J$ with $\mu(\cF)>0$. The two regions are separated by
a stability wall in K\"ahler moduli space where $\mu(\cF)=0$ and $V$ is semi-stable. Using
the short exact sequence
\beq\label{ferret}
 0 \to \cF \to V \to V/\cF\to 0\; ,
\eeq
we note, as in Sections \ref{bundlestab} and \ref{eft}, that we can write $V = {\cal F} \oplus V/\cF$ as an element in its S-equivalence class.

Our physical four-dimensional picture of bundle stability suggests certain conditions on bundle cohomology which we now discuss. Due to the Fayet-Iliopoulos (FI) D-term \eref{thedterm} derived in this paper, the preservation of supersymmetry in the effective theory depends upon the existence (or absence) of certain charged matter fields (the fields $C^L$ in \eref{thedterm}) described by $H^1(X, \cF \otimes (V/\cF)^*)$ and $H^1(X, \cF^* \otimes  V/\cF)$. Specifically, in order to preserve supersymmetry in the region of moduli space with $\mu(\cF)<0$, the fields $C^L$ described by $H^1(X, \cF \otimes (V/\cF)^*)$ must acquire a vacuum expectation value and cancel the FI term in \eref{thedterm}, hence setting the potential to zero in this region of K\"ahler moduli space. In particular, this means that such fields must exist and hence $H^1(X, \cF \otimes (V/\cF)^*)\neq 0$. On the other hand, if the region of moduli space for which  $\mu(\cF)>0$ is to have broken supersymmetry, there must be {\it no} fields $C^L$ described by $H^1(X, \cF^* \otimes V/\cF)=0$. Stating this more formally, we must have the following Lemma:

\begin{lemm}\label{vanishing}
Let $V$ be a holomorphic vector bundle with structure group $SU(n)$ ($n=3,4,5$) defined over $X$, a Calabi-Yau $3$-fold with K\"ahler form $J$. If $\cF$ is a rank $(n-1)$, stable sub-sheaf of $V$, defining a ``stability wall" in in the K\"ahler cone given by $\mu(\cF)=0$, such that $V$ is stable for $\mu(\cF)<0$ and unstable for $\mu(\cF)>0$, then $H^1(X, \cF \otimes (V/\cF)^*) \neq 0$ and $H^1(X, \cF^* \otimes V/\cF)=0$ (for any effective field theory describing $V$).
\end{lemm}
\begin{proof} We begin with the first condition $H^1(X, \cF\otimes (V/\cF)^*) \neq 0$.
Consider twisting the sequence \eref{ferret} by the
line bundle ${\cal K}^*$, where ${\cal K}=(V/\cF)^{**} \approx V/\cF$. This leads to the short exact sequence
\beq 0 \to \cF\otimes {\cal K}^* \to V\otimes {\cal K}^* \to {\cal
  K}\otimes {\cal K}^* \to 0\; .
\eeq
Then the associated long exact sequence in
cohomology contains the terms
\beq 0 \to H^0(X, \cF\otimes {\cal K}^*)
\to H^0(X,V\otimes {\cal K}^*) \to H^0(X, {\cal K}\otimes {\cal K}^*)
\to H^1(X,\cF\otimes {\cal K}^*) \to \ldots
\eeq

Because $V$ is stable for $\mu(\cF)<0$, it must follow that at a generic point in the bundle moduli space of $V$, $H^0(X,V\otimes {\cal K}^*)=0$ (otherwise ${\cal K}$ would be a sub-sheaf of $V$ and would destabilize $V$). Furthermore, since ${\cal K}$ is a line-bundle on a Calabi-Yau manifold,  $H^0(X, {\cal K}\otimes {\cal K}^*)=1$. As a result, we have \beq 0 \to H^0(X, {\cal K}\otimes
{\cal K}^*) \to H^1(X,\cF\otimes {\cal K}^*) \to \ldots \eeq and it is
clear that we {\it must} have $H^1(X,\cF\otimes {\cal K}^*)\neq
0$ in order to avoid a contradiction. However, since the value of this cohomology is unaffected as we
move to the decomposable locus in the moduli space of $V$ (as
described in Section \ref{bundlestab}), we see that $H^1(X,\cF\otimes
(V/\cF)^*)\neq 0$ is satisfied, as expected.

We turn now to the second cohomology condition that we must investigate. In order for the theory to break supersymmetry above the line with $\mu(\cF)=0$, it must be the case that $H^1(X, \cF^* \otimes V/\cF)=0$. This too follows immediately from the definition of a stability boundary. Suppose that $H^1(X, \cF^* \otimes V/\cF) \neq 0$, then there exists a non-trivial extension:
\beq
0 \to V/\cF \to \tilde{V} \to \cF \to 0~.
\eeq
But, by definition, this implies that there exists an injective map from $V/\cF$ to $\tilde{V}$ at generic points in moduli space (away from the decomposable locus). Thus, we must ask, can $\tilde{V}$ be isomorphic to $V$? If this is the case, then $V/\cF$ is a sub-sheaf of $V$ that destabilizes $V$ in the region $\mu(\cF)<0$. But by construction, we know that $\cF$ destabilizes $V$ in the region with $\mu(\cF)>0$, hence the bundle is stable nowhere in K\"ahler moduli space. This is a contradiction, since we are considering the situation in which the line $\mu(\cF)=0$ defines the boundary of a stable/unstable transition in the moduli space. Thus, if $V$ is stable for $\mu(\cF)<0$, then the extension $Ext^1(V/\cF, \cF)=H^1(X, \cF^* \otimes V/\cF)$ defining $\tilde{V}$ is {\it not} isomorphic to $V$. Thus, if $H^1(X, \cF^* \otimes V/\cF) \neq 0$ we are considering an effective theory in which branch structure is present, connecting more than one vector bundle. However, for the statement of this lemma, we shall consider only the effective theory describing $V$. \end{proof}

Note that similar vanishing conjectures could be formulated for other
possible bundle decompositions of $V$. However, in the presence
of additional $U(1)$ gauge fields and different decompositions of
$E_8$ under these symmetries, each case must be investigated on an
individual basis.

Next, we turn to the proof of the second lemma in Section \ref{comparison}. It states the following:
\begin{mma}
Let $V$ be a holomorphic vector bundle with structure group $SU(n)$ defined over $X$, a Calabi-Yau $3$-fold with K\"ahler form $J$. If $\cF$ is a rank $n-1$, stable sub-sheaf of $V$, defining a stability wall in in the K\"ahler cone given by $\mu(\cF)=0$, such that $V$ is stable for $\mu(\cF)<0$ and unstable for $\mu(\cF)>0$, and $H^1(X, \cF^* \otimes V/\cF)=0$, then 
\beq h^1(X,V \otimes V^*)=h^1(X, \cF \otimes (V/\cF)^*)+h^1(X, \cF \otimes\cF^* )-1~,\eeq 
where $h^1(X,V\otimes V^*)$ is the generic dimension of bundle moduli space when $V$ is a stable bundle.\end{mma}
\begin{proof}
Consider once again the short exact sequence \eref{ferret} which defines the sub-sheaf $\cF$. In order to relate the generic (stable) bundle moduli of $V$ to the possible deformations of $\cF \oplus {\cal K}$, we will compute $h^1(X, V \otimes V^*)$ using \eref{ferret}. To begin, we consider the following three short exact sequences that follow directly from \eref{ferret}.
\bea
0 \to \cF \otimes V^* \to V \otimes V^* \to {\cal K} \otimes V^* \to 0\label{modu1} \\
0 \to\cF\otimes {\cal K}^* \to \cF\otimes V^* \to \cF \otimes \cF^* \to 0\label{modu2} \\
0 \to {\cal K}\otimes {\cal K}^* \to {\cal K}\otimes V^* \to {\cal K}\otimes \cF^* \to 0\label{modu3} 
\eea
From these sequences we can consider long exact sequences in cohomology. We begin with \eref{modu2}. Using the results of Lemma I, and the fact that for this class of examples $V$ is stable for $\mu(\cF)<0$, we have $H^0(X,\cF \otimes V^*)=0$ and $H^2(X, \cF \otimes {\cal K}^*)=0$. Thus,
\beq
0 \to H^0(X, \cF \otimes \cF^*) \to H^1(X,\cF \otimes {\cal K}^*) \to H^1(X, \cF \otimes V^*) \to  H^0(X, \cF \otimes \cF^*) \to 0~.
\eeq
Next, from \eref{modu3}, we note that since ${\cal K}$ is a line bundle, ${\cal K} \otimes {\cal K}^* \approx \cO$ and hence, $H^1(X, {\cal K} \otimes {\cal K}^*)=0$ and by Lemma I, we have that $H^1(X, {\cal K} \otimes \cF^*)=0$. Further, we have $H^0(X, {\cal K}\otimes \cF^*)=0$ since $\cF$ is stable. Hence, it follows that 
\beq
h^0(X, {\cal K} \otimes V^*)=1~~\text{and}~~h^1(X, {\cal K}\otimes V^*)=0~.
\eeq 
Substituting this information into the cohomology sequence for \eref{modu1}, we find
\beq
0\to H^0(X, V\otimes V^*) \to H^0(X, {\cal K}\otimes V^*) \to H^1(X, \cF \otimes V^*) \to H^1(X, V\otimes V^*) \to 0
\eeq
Then, in terms of dimensions:
\beq
h^1(X, V\otimes V^*)= h^0(X, V\otimes V^*) -h^0(X, {\cal K} \otimes V^*) +h^1(X, \cF \otimes V^*)
\eeq
and upon substitution
\beq
h^0(X, V\otimes V^*) -1+h^0(X, \cF\otimes\cF) +h^1(\cF\otimes\cF)+h^1(X, \cF\otimes{\cal K}^*)~.
\eeq
Finally, since $V$ and $\cF$ are stable $h^0(X, V\otimes V^*)=1=h^1(X, \cF\otimes\cF^*)$ and we arrive at the result,
\beq
h^1(X, V\otimes V^*)=h^1(X, \cF\otimes{\cal K}^*)+h^1(\cF\otimes\cF^*)-1
\eeq
as required.
\end{proof}

\vspace{0.5cm}

We end this section by a statement of a conjecture. This is less easy
to verify than the cohomology conditions described above, though we
have found it to be true in all the cases that we have
investigated. The central result of this paper is the form of the $FI$
D-term given in \eref{thedterm} which reproduces the notion of vector
bundle stability for a supersymmetric, anomaly
free\footnote{Recall that a vector bundle $V$ in the $E_8 \times E_8$
  heterotic theory defines an anomaly free superymmetric theory if
  $ch_{2}(TX)-ch_{2}(V)=W$ where $W$ is an effective class of $X$ \cite{Green:1987mn,bogomolov, Douglas:2006jp}. This condition is necessary here,
  as without it anti five-branes or a non-supersymmetric hidden bundle
  would be required to make the reduction from eleven to four
  dimensions consistent. This would result in a theory which was not
  supersymmetric in four dimensions \cite{Gray:2007qy,Gray:2007zza},
  and as such the analysis of this paper would not apply.} bundle. As
discussed in Section \ref{comparison}, the form of this potential
clearly does not depend on the complex structure moduli of the Calabi-Yau manifold $X$.
In addition, using the techniques of Section \ref{bundlestab}, we have searched through numerous examples, and have
yet to find a complex structure dependent boundary wall for an anomaly
free bundle. As a result, we posit the conjecture:

\begin{conjecture}
Let $V$ be an anomaly-free holomorphic vector bundle with structure group $SU(n)$ ($n=3,4,5$) defined over $X$, a Calabi-Yau $3$-fold. If there exists a wall of semi-stability of $V$ in K\"ahler moduli space (defining the boundary between stable and unstable regions), then the position of this wall is independent of the complex structure moduli of $X$. 
\end{conjecture}

This conjecture is a consequence of our field-theoretical approach to slope-stability but it is not obvious to the authors how to prove it from an algebraic geometry viewpoint.

\section{Appendix B: Another example}\label{AppendixB}

To highlight the versatility of the formalism developed in this paper, in this section we will sketch another example bundle, its regions of stability in the K\"ahler cone and the effective field theory modeling this behavior.

We shall once again consider a bundle defined on a complete intersection Calabi-Yau manifold, $X$. The so-called `bi-cubic' $3$-fold:
\begin{equation}
X= \left[\begin{array}[c]{c}\mathbb{P}^2\\\mathbb{P}^2\end{array}
\left|\begin{array}[c]{ccc}3 \\3
\end{array}
\right.  \right]\; ,
 \label{cicy33}
\end{equation}
defined by a polynomial of bi-degree $(3,3)$ in the ambient space $\mathbb{P}^2\times\mathbb{P}^2$. 
As in our previous example, $h^{1,1}(X)=2$ and the K\"ahler cone is the positive quadrant $t^1\geq 0$ and $t^2\geq 0$. The non-zero triple intersection numbers are given by $d_{122}=3$ and $d_{112}=3$. It follows that the dual K\"ahler moduli $s_1$ and $s_2$ are
\begin{equation}
 s_1=3t^2(2t^1+t^2)\; ,\quad s_2=3t^1(2t^2+t^1)\; . \label{dkah}
\end{equation} 
Hence, the dual K\"ahler cone in this case is the entire positive quadrant. We shall define line bundles $\cO_X(m,n)$ on this space using the same notation as in Section \ref{simpleeg}.

On the given manifold, we define a bundle, $V$, by extension \cite{burt},
\beq\label{extensioneg}
0 \to W \to V \to {\cal L} \to 0~,
\eeq
where ${\cal L}$ is a line bundle and $W$ is a rank 2, $U(2)$ monad bundle defined as follows
\bea\label{Wdef}
0\to W \to \cO_X(2,0)^{\oplus 3} \to \cO_X(2,2) \to 0 \\
{\cal L}=\cO_X(-4,2).~~~~~~~~~~~~~~\nn
\eea 
Since the first chern classes of $\cL$ and $W$ satisfy
\beq
c_1(W)=-c_1(\cL)
\eeq
the extension bundle $V$ defined by \eref{extensioneg} has $c_1(V)=0$ and hence defines an $SU(3)$ bundle. Furthermore, $V$ is a non-trivial extension of $\cL$ by $W$ (i.e. $V$ is not simply the sum $W\oplus \cL$ since $Ext^1(\cL, W) \neq 0$).  The spectrum of the four dimensional $E_6$ theory associated to $V$ consists of $18$ ${\bf 27}$ matter fields and $18$ $\overline{{\bf  27}}$'s for a net chiral asymmetry of zero. In addition, there are generically $h^1(X,V\otimes V^*)=530$ bundle moduli.

We can now ask, what are the regions of stability of $V$ in the K\"ahler cone? A simple analysis using the techniques of Section \ref{constl} verifies first that $W$ is an everywhere stable $U(2)$ bundle, and furthermore, that $W$ is generically the {\it only} de-stabilizing sub-sheaf of $V$. Thus, since $c_1(W)=4J_1-2J_2$, $V$ itself is stable above the line with slope $s_2/s_1=2$ and unstable beneath it. We will now reproduce this geometric result from the point of view of the effective field theory developed in this work.

As was argued in Section \ref{bundlestab}, at the line of semi-stablility in the dual K\"ahler cone defined by $s_2=2s_1$, $V$ will be forced away from an $SU(3)$ configuration towards the structure group $S(U(2)\times U(1))$ (and the four dimensional symmetry will be enhanced to $E_6 \times U(1)$). As in the example given in \ref{thefirstexample}, $V$ decomposes as  $V={\cal F}\oplus{\cal K}$ where in this case ${\cal F}=W$ and ${\cal K}=\cL$ (as defined above in \eref{Wdef}). Note that the split locus is simply the zero of the group $Ext^1(\cL,W)$ which describes the space of possible extensions. Using the results of \cite{Anderson:2007nc,Anderson:2008ex,Anderson:2008uw} to compute the cohomology of $W$ and $\cL$ on the bi-cubic, and the representation decomposition given in Table \ref{table1}, we find that non-vanishing massless spectrum of $V$ at the decomposable locus is given by
\bea\label{nvspec}
h^1(X,W)_{-1/2}=18~~~~~h^1(X,\cL^*)_{-1}=18~,~~~~~~~\\
h^1(X, W\otimes W)_{0}=9~~~h^1(X,W\otimes \cL^*)_{-3/2}=522~.~~~
\eea
The subscript on the cohomology denotes the $U(1)$ charge of the fields. We may now write down the $U(1)$ D-term \eref{thedterm} contribution to the potential

\begin{equation}
\label{dtermeg2}
D^{U(1)} = \frac{3}{16} \frac{\epsilon_S \epsilon_R^2}{\kappa_4^2}\frac{\mu (W)}{\cal V} +\frac{3}{2} \sum_{L,\bar{M}=1}^{16}G_{L \bar{M}} C^L \bar{C}^{\bar{M}}\; ,
\end{equation}
where here, from \eref{Wdef}, the slope is given by
\begin{equation}
 \mu(W)=\frac{1}{2}(4s_1-2s_2)\;
\end{equation} 
while for the volume we have
\begin{equation}
 {\cal V}=\frac{3}{2}(t^1t^2)(t^1+t^2)\; .
\end{equation} 

The relevant charged matter fields $C^L$ in this case are the fields in $H^1(X, W \otimes \cL^*)$, since the fields associated to $H^1(X,W)$ and $H^1(X,\cL^*)$ will have vevs forced to zero by the requirement that $E_6$ remains unbroken. As we would predict based from the algebro-geometric results of the stability analysis, negatively charged matter is present so that the vevs of the charged fields can adjust to cancel the FI term when $\mu(\cF)<0$, setting the D-term to zero. Thus, supersymmetry is preserved for the region of dual K\"ahler moduli space defined by $s_2>2s_1$. However, since there is no positively charged matter available, for the region of moduli space where $\mu(\cF)>0$, the FI term cannot be cancelled and supersymmetry is broken. This is in agreement with what we would expect from the general results of Lemma \ref{lemma_1}.

Finally, for this example, we can verify the general predictions of Lemma \ref{lemma_2} by considering the number of bundle moduli associated to $V$ at a generic point in its moduli space as well as at the decomposable locus. According to Lemma 2, we would expect there to be one extra light modulus at the stability wall. For the bundle $V$ defined by \eref{extensioneg}, at a generic point in its moduli space, $h^1(X,V\otimes V^*)=530$. Moreover, using the results of \eref{nvspec} we observe that at the decomposable locus, the number of bundle moduli is given by $h^1(X,W\otimes W^*) +h^1(X, W\otimes \cL^*)=531$. Thus, as described in Section \ref{comparison}, as we move in K\"ahler moduli space away from the stability wall, one degree of freedom is made massive by the Higgs mechanism,  \eref{u1mass}, as expected.


\end{document}